\documentclass[11pt]{amsart}

\usepackage{hyperref}
\hypersetup{
    colorlinks=true,
    linkcolor=blue,
    filecolor=magenta,      
    urlcolor=blue,
}

\usepackage{color}

\usepackage{latexsym}
\usepackage{epsfig,graphics}
\usepackage{amsmath,amssymb,amsthm}
\usepackage{graphics,epsfig,calc}
\usepackage{bm}

\DeclareMathAlphabet\mathbfcal{OMS}{cmsy}{b}{n}

\usepackage{upref}
\usepackage{setspace}
 \newcommand{\beqn}{\begin{eqnarray}}
 \newcommand{\eeqn}{\end{eqnarray}}
 \newcommand{\be}{\begin{equation}}
 \newcommand{\ee}{\end{equation}}
 \newcommand{\ba}{\begin{array}}
 \newcommand{\ea}{\end{array}}

 \newcommand{\pa}{\partial}
 
  \newcommand{\ci}{\cite}
 \newcommand{\ds}{\displaystyle}
 \newcommand{\la}{\label}

  \newcommand{\llangle}{{\langle\!\langle}}
  \newcommand{\rrangle}{{\rangle\!\rangle}}

   \newcommand{\lan}{\langle}
    \newcommand{\ran}{\rangle}
      \newcommand{\un}{ | \hat u_n|^2}
       \newcommand{\strela}{\rightharpoonup}

 \newcommand{\fr}{\frac}

\newcommand{\eps}{\epsilon}
\newcommand{\ov}{\overline}

\newcommand{\lann}{\langle\langle}
\newcommand{\ra}{\rangle }
\newcommand{\rann}{\rangle\rangle }

\newcommand{\cF}{ \mathcal{F} }

\newcommand{\cB}{{\mathcal B}}
\newcommand{\cD}{{\mathcal D}}

\newcommand{\aE}{{\mathbb{E}}}
\newcommand{\EE}{{\mathbb{E}}}
\newcommand{\PP}{{\bf P}}

\newcommand{\aP}{{\mathbb{P}}}

\newcommand{\cH}{{\mathcal H}}

\newcommand{\cM}{{\mathcal M}}

\newcommand{\cP}{{\mathcal P}}

\newcommand{\vp}{\varphi}
\newcommand{\De}{\Delta}
\newcommand{\de}{\delta}

\newcommand{\dist}{{\rm dist\5}}

\newcommand{\xinn}{{}^{N}\!\xi}
\newcommand{\xin}{\xi^N}

\newcommand{\al}{\alpha}

\newcommand{\vka}{{\varkappa}}
\newcommand{\si}{\sigma}
\newcommand{\om}{\omega}

\newcommand{\Om}{\Omega}
\newcommand{\na}{\nabla}

\newcommand{\Si}{\Sigma}
\newcommand{\lam}{\lambda}

\newcommand{\5}{{\hspace{0.5mm}}}
\newcommand{\N}{\mathbb{N}}
\newcommand{\R}{\mathbb{R}}
\newcommand{\Z}{\mathbb{Z}}

\newtheorem{theorem}{Theorem}[section]

\newtheorem{defin}[theorem]{Definition}

\newtheorem{lemma}[theorem]{Lemma}
\newtheorem{remark}[theorem]{Remark}
\newtheorem{remarks}[theorem]{Remarks}
\newtheorem{cor}[theorem]{Corollary}
\newtheorem{pro}[theorem]{Proposition}
\newtheorem{exe}[theorem]{Example}
\newtheorem{exes}[theorem]{Examples}
\newtheorem{exer}[theorem]{Exercise}

\newcommand{\bp}{\begin{pro}}
\newcommand{\ep}{\end{pro}}

\newcommand{\bt}{\begin{theorem}}
\newcommand{\et}{\end{theorem}}

\newcommand{\bl}{\begin{lemma}}
\newcommand{\el}{\end{lemma}}

\newcommand{\sgn}{\mathop{\rm sgn}\nolimits}

\newcommand{\bce}{\begin{center}}
\newcommand{\ece}{\end{center}}

\newcommand{\bpr}{\begin{proof}}
\newcommand{\epr}{\end{proof}}

\newcommand{\br}{\begin{remark}}
\newcommand{\er}{\end{remark}}

\newcommand{\brs}{\begin{remarks}}
\newcommand{\ers}{\end{remarks}}

\newcommand{\bd}{\begin{defin}}
\newcommand{\ed}{\end{defin}}

\newcommand{\bc}{\begin{cor}}
\newcommand{\ec}{\end{cor}}

\newcommand{\bex}{\begin{exe}}
\newcommand{\eex}{\end{exe}}
\newcommand{\bexs}{\begin{exes}}
\newcommand{\eexs}{\end{exes}}
\newcommand{\bexe}{\begin{exer}}
\newcommand{\eexe}{\end{exer}}
\author{Sergei Kuksin}
\address
{Universit\'e  Paris Cit\'e and Sorbonne Universit\'e, CNRS, IMJ-PRG, F-75013 Paris, France, and
 Peoples' Friendship University of Russia,  
 Moscow,  Russia, and Steklov Mathematical Institute of
  Russian Academy of Sciences, Moscow, Russia, and  Institute for Financial Studies, Shandong University, Jinan, China.
}
\email{sergei.kuksin@imj-prg.fr}

\title
{  Stochastic 1d Burgers equation as a model for hydrodynamical turbulence. 	}

\begin{document}
\date{}

\maketitle

\begin{abstract}
This work is a review with proofs of a group of results on  the stochastic  Burgers equation with small viscosity, 
 obtained during the last two decades. These results   jointly show that the 
 equation makes a surprisingly good model of hydrodynamical turbulence. The model provides natural and
rigorously justified analogies of a number of key predictions of the theory of turbulence, including the main 
assertions of the Kolmogorov approach to turbulence, known as the K41 theory. 
\end{abstract}


\tableofcontents
\setcounter{equation}{0}
\section {Introduction: Kolmogorov's theory and its 1d model}

A theory of turbulence, known now as the K41 theory, was created 
in three  articles, published by  A.N.~Kolmogorov  in 1941, and in two articles of his student  Obukhov which appeared the  same year. 
Probably now this is the most popular theory of turbulence, but as all other theories of hydrodynamical turbulence it is heuristic, and it is unclear if in the foreseeable future its claims will be rigorously justified. So  at the 
current stage of development of the field any mathematically correct theory which is consistently 
related to K41 and may  be compared with it is useful and important. 
 In this paper we give a concise survey of such  a theory  which deals with  turbulence in fictitious one-dimensional fluid, 
  described by the space-periodic one-dimensional   stochastic  Burgers equation. 

The Burgers equation as a 1d model of fluid motion was suggested by Burgers in late 1930's  and since then was  
systematically used in this quality by him (e.g. see \cite{B1948}) and by other 
 experts in hydrodynamics. In 1980's-1990's  Frisch  studied   on the physical level of rigour the 
equation with small viscosity and random initial data and/or random forcing, regarding this as a stochastic model of 1d turbulence, 
 see \cite{FF, AFLV}. Motivated by this work, in 1990's Sinai,  himself and with collaborators, 
   started to examine the stochastic Burgers equation under the inviscid limit $\nu\to0$. This research resulted in 
the influential paper \cite{WKMS} and  then  was continued by 
Sinai's  students and followers (including {E,  Iturriga, Khanin and others; see \cite{IK2003, KZ} and references in these works).  Early 
this century, also in connection with 1d turbulence,  the space-periodic 
  Burgers equation with small positive viscosity was examined  by two students of the  author,  Biriuk  and Boritchev,
  using tools from 
nonlinear PDEs and some ideas from previous work of the author on nonlinear PDEs with small dissipation (see \cite{BK} for references). 
    The study was continued by the two and the author and   led to  the book
\cite{BK}, which shows  that  many basic statements of Kolmogorov's  theory 
allow a rigorous interpretation in the Burgers framework in terms of  solutions  for  the stochastic Burgers equation.\footnote{We also mention
the work  \cite{CGR}, where 
 without any relation to the Burgers equation is constructed a class of stationary processes $u^\nu(x)$, $x\in S^1$,
  satisfying  1d versions of the main predictions of K41.}

The goal of this paper is to present  main results of the book \cite{BK} and some their recent developments in a  form, lighter than in  \cite{BK},  
and hopefully  more suitable
for physicists and those mathematical readers who are less concerned with the  rigour of  argument.  The lightness of the presentation 
is achieved by giving  without verification  a few lemmas which  we regard as technical, less interesting and less important 
  (their demonstrations may be found in \cite{BK}), and  by 
   significant   shortening  the proofs: we    assume that the omitted details do not interest  readers from physics 
     and may be   relatively easily recovered  by those from mathematics. 
     Besides,   we omit some results from \cite{BK}   which regard as less important. 
  At the origin of this work  lies the lecture notes  of an online course which the author was teaching 
    at the Fudan University in December of 2021.

Now we will briefly  describe the content of  our work.
Short Sections~\ref{s_2}-\ref{s_5} are preliminary. There we develop some analysis, needed for the main part of the paper, show 
that the space-periodic stochastic Burgers equation is well-posed in Sobolev spaces and defines there Markov processes.  
Then in Sections~\ref{s_olei} and \ref{s_8} we discuss the behaviour of solutions for the equation as the 
viscosity $\nu$ goes to zero. In Sections~\ref{s_burgulence}--\ref{s_45} we talk 
about  properties of solutions for the Burgers equation with small viscosity in parallel  with the assertions of the K41 theory,
regarding the derived there  results  as the laws of turbulent  motion of  fictitious 1d ``burgers fluid", i.e.  as the laws of 1d turbulence.
Finally in Section~\ref{s_9}    we discuss the well known (e.g. see \cite{WKMS}) 
 existence of an inviscid   limit for solutions of  the  Burgers equation 
as $\nu\to 0$: $u^\nu(t,x)\to u^0(t,x)$, a.s.  The limit $u^0$  is a discontinuous function, bounded for bounded $t$, which 
 satisfies the inviscid   stochastic Burgers  equation in the sense of generalised functions, and is traditionally called an 
  {\it entropy}, or an {\it inviscid}    solution. Passing to the limit in the results of Sections~\ref{s_burgulence}--\ref{s_45}  we show 
  that the entropy solutions $u^0$ possess a collection of properties which with good reasons may be called {\it inviscid 1d
  turbulence}. 
  \smallskip
  
  Most of the results on 1d turbulence whose rigorous proof we talk 	about in this work were earlier obtained on heuristic level of 
  rigour by Burgers himself and in works of other physicists, e.g. in \cite{AFLV}. We discuss this in the main part of the paper.

Below in our work,  speaking
  about the K41 theory,  as in the K41 papers we usually assume that the turbulent velocity fields $u(t,x)$ under discussion are random fields,
  stationary in time and homogeneous in space. In addition, unless otherwise stated, 
   we suppose that these fields are space-periodic, and normalise   the period to be one.
  We assume that  the units are  chosen in such a way that the velocity  fields are of order one, uniformly in small $\nu$:
  \be\label{Kenergy}
  \EE | u(t,x)|^2 \sim1.
  \ee
   Then their Reynolds numbers  equal $\nu^{-1}$, where $\nu$ is the viscosity. Moreover, as in K41 we 
  suppose that the rates of energy dissipation $\eps^K$ of the flows
   remains of order one  as $\nu\to0$, 
 \be\label{Krate}
 \eps^K := \nu \EE |\nabla u(t,x)|^2 \sim1. 
 \ee
 Here and below  $A\sim B$ means  that  ratio $A/B$   is bounded from below and from 
       above by positive constants, independent of  $\nu$ and of  the indicated arguments in $A$ and $B$.

\noindent{\bf Agreements and Notation. } In our work 
all random processes have continuous trajectories, and  always  the process $\xi$  in relations  \eqref{xi}
is that  introduced in Proposition~\ref{p1}.
For a Banach space $X$  and $R>0$ we denote by $B_X^R$ the open ball 
$
\{u\in X:  \|u\|_X<R\},
$
and  by  $\ov B_X^R$ -- its closure. By $\| \cdot\|_m$ we  denote the homogeneous $m$-th
Sobolev norm for functions on $S^1$ (see \eqref{nHm}) and by $| \cdot|_p$ -- the norm in $L_p(S^1)$. 
By $H^m$, $ m\in\R$, we denote the Sobolev space of order $m$ of functions on  $S^1$ with zero mean value. 
Any
metric space $M$ is provided
with the Borel $\si$-algebra $\cB(M)$ (see \cite{Sh} and 
 Appendix~C in \ci{BK}). So when we say that  ``a map to $M$ is measurable",
 it means that it is  measurable with respect to $\cB(M)$. By $\cP(M)$ we denote 
 the set of probability Borel  measures on $M$, and by the symbol $\strela$ denote the weak convergence of measures. 
   A set from $\cB(M)$ of zero measure is called a null-set. 
  For a function $f$ and a measure $\mu$ we  write
$$
\langle f,\mu\rangle= 
\int f(u)\mu(du)
$$
(a clash of notation with the $L_2$-scalar product should not cause a problem for the reader).

\section {The  setting }\label{s_2}
\subsection {The Burgers equation}

The initial-value problem for the  space-periodic Burgers equation reads
\be\la{B}
\left\{\ba{rcl}
u_t(t,x)+u u_x-\nu u_{xx}&=&\eta(t,x),\qquad t\ge 0,
\\\
u(0,x)&=&u_0(x).
\ea\right|\,\,x\in S^1=\R/\Z.
\ee
 Here 
$$
0<\nu \le1,
$$
so everywhere below ``for all $\nu$" means ``for all $0<\nu\le1$". 
The force $\eta$ is a random field $\eta^\om(t,x)$, 
defined on a probability space $(\Om,\cF,P)$, and specified below. 
All   details
on probability objects and assertions  which are given below without explanation may be found  e.g.  in \ci{Sh}
and  appendices to \cite{BK}.

We always assume that
$
\int \eta^\om(t,x)dx\equiv\int u^\om_0(x)dx=0.
$ 
Since $u u_x=\fr12 \fr\pa{\pa x} u^2$, then integrating the Burgers equation (\ref{B}) over $S^1$, we get that
$
\fr\pa{\pa t}\int u(t,x)dx\equiv 0, 
$
so
$$
\int u(t,x)dx\equiv 0,\qquad t\ge 0.
$$
Consider the space 
$$
H=\{u\in L_2(S^1):\int u(x)dx=0\}, 
$$
equipped with the $L_2$-scalar product $\langle \cdot, \cdot \ra$ and the $L_2$-norm $\| \cdot\|$ (so $\| u\|^2 = \langle u, u \ra$).  We will 
 regard a solution $u$ of the Burgers equation (\ref{B})
either as a function $u(t,x)$ of $t,x$, or as a curve $t\mapsto u(t,\cdot)=:u(t)\in H$,
depending on the random parameter $\om$. That is, either as a random field $u^\om(t,x)$, or as a random process $u^\om(t) \in H$. 

Below in Sections \ref{s_3}-\ref{s_8} we show that eq.~\eqref{B} is well posed and study properties of its solutions with small $\nu$. In particular,
we obtain lower and upper bounds for second moments of their Sobolev norms which are asymptotically sharp as $\nu\to0$ in the sense that
they involve $\nu$ in the same negative degree. Then in Section~\ref{s_burgulence} we state one-dimensional versions of the main laws of the K41 
theory and use results of the previous sections to prove them rigorously for the fictitious  1d fluid whose motion is described by eq.~\eqref{B}.

\subsection{Function spaces and random force}\label{s_force}

 We denote by  $\{e_s(x)\in H: s\in\Z^*=\Z\setminus \{0\}\,\}$ the orthonormal trigonometric  basis of $H$, 
\be\la{baH}
e_s(x)=\left\{\ba{ll}
\sqrt{2}\cos 2\pi sx,&s\ge 1,\\
\sqrt{2}\sin 2\pi |s|x,&s\le -1.\ea \right .
\ee
Any $u\in H$ decomposes as
$$
u(x)=\sum_{s\in \Z^*}  u_s e_s(x),\qquad x\in S^1,
$$
and  may be written as  Fourier series
$$
u(x)=\sum  \hat u_s e^{2\pi i sx},\qquad \hat u_s= \ov{\hat u}_{-s}=
\fr1{\sqrt{2}}(u_s-iu_{-s}),\quad s\in\N;\quad \hat u_0=0.
$$

\noindent
{\it The force and solutions.} We suppose that $\eta(t,x)$ is a regular function of $x$, while as a function of $t$ it is a
 distribution:
\be\la{xi}
\eta=\eta^\om(t,x)=\pa_t\xi^\om(t,x), \qquad 
\xi^\om(t,x)=\sum_{s\in Z^*} b_s\beta_s^\om(t)e_s(x).
\ee
Here $\{b_s\}$ are real numbers, and $\{\beta_s\}$ are standard independent
Wiener processes  on  $(\Om,\cF,P)$. Abusing language we also call the random field 
 $\xi$ ``a force". 
 It is easy to see that 
\be\la{hom}
\text{if  $b_s \equiv b_{-s}$, then the random field $\xi(t,x)$ is homogeneous in $x$} 
\ee
(see \cite[Section 1.5]{BK}).

For 
$m\in\N_0=\N\cup 0$ we denote
$$
B_m=\sum |2\pi s|^{2m}b_s^2\le \infty\,,
$$
and will   always assume that
\be\label{B_assume}
B_0>0, \qquad 
B_{m}<\infty  \;\; \forall\, m.
\ee
The first relation in \eqref{B_assume} is needed for majority of our results, while the second may be weakened, see in \cite{BK}.

Let $m\in\N$. 
The Hilbert space $H^m$ is the Sobolev  space $\{v\in H: ({\partial}^m/{\partial} x^m) v  \in H\}$, equipped 
 with the homogeneous  Hilbert norm
\be\la{nHm}
\Vert v\Vert_m:= \|  \tfrac{\partial^m}{\partial x^m} v  \|.
\ee
If $v(x)=\sum v_s e_s(x)$, then
$
\Vert v\Vert_m^2 =\sum |2\pi s|^{2m}|v_s|^2.
$ 
By this relation  we define the norm $\Vert v\Vert_m$
for any $m\in\R$. Then, for $m\ge 0$,
$\,
H^m:=\{v\in H: \Vert v\Vert_m<\infty\},
$ 
and $H^{-m}$ is the complement of $H$ in the norm $\Vert\cdot\Vert_{-m}$. We also set
$
H^\infty = \cap H^m = C^\infty(S^1) \cap H.
$
 Next, for 
 $0<T<\infty$,  we introduce  the Banach spaces 
$$
X_T^m=C(0,T; H^m)\,. 
$$
For   $T=\infty$ we set 
$
X^m_\infty=C(0,\infty; H^m)
$. 
This  is a complete separable metric space with the distance
$$
\dist(u,v)=\sum_{T=1}^\infty 2^{-T} \psi\big( |(u-v)\!\mid_{[0,T]} |_{X_T^m}\big),\qquad \psi(r):=r/(1+r),\,\,\,r\ge 0.
$$

Well known basic properties of the random field  $\xi^\om(t,x)$  in \eqref{xi} are described by the following  proposition (e.g. see 
 \ci{BK} for a  proof).

\bp\la{p1}
 If \eqref{B_assume} holds, then  there exists a 
  null-set $Q$ such that for each non-negative integer $m$ we have:

\noindent 
i) for  $\om\not\in Q$ and $t\ge0 $
the series in (\ref{xi}) converges in $H^m$ to a limit  $\xi (t)$ which is a continuous process in $H^m$, vanishing at zero. 
 For $\om\in Q$ we set $\xi=0$. 
\smallskip\\
ii) $\aE\Vert\xi(t)\Vert_m^2 = tB_m\quad \forall\,t\ge0$.
\smallskip\\
iii) For any $T<\infty$,\,
$\aE e^{\al\Vert\xi\Vert_{X_T^m}^2} \le  4 e^{2T\al B_m}-3
\quad{\rm if}\;\; \al\le 
\al_m(T)= 1/(4TB_m).
$
\ep

The process $\xi(t)$ in this proposition is {\it a Wiener process in the spaces $H^m$.}

\setcounter{equation}{0}
\section {Deterministic equation} \label{s_3}
Consider first  the initial-value problem for  the deterministic  Burgers equation 
\be\la{B2}
\left\{\ba{rcl}
u_t(t,x)+u u_x-\nu u_{xx}&=&\eta(t,x)=\pa_t \xi(t,x),\qquad t\ge 0,
\\\
u(0,x)&=&u_0(x)
\ea\right|,\,\,x\in S^1=\R/\Z,
\ee
where $\xi \in  C([0,\infty),H)$ and $u_0\in H$ are non-random. We say that 
a function $u(t,x) \in C([0,\infty),H)$ solves  (\ref{B2})  if 
\be\la{dB}
u(t,x)-u_0(x)+\int_0^t [u(s,x)u_x(s,x)-\nu u_{xx}(s,x)]ds=\xi(t,x)-\xi(0,x), 
\ee
for all $t\ge 0$. Since $u(s, \cdot) \in H$ and $uu_x = \frac12 (\pa/\pa x) u^2$,  then for any $t$ the l.h.s. of \eqref{dB}
as a function of $x$  is a
well defined distribution, and equality \eqref{dB} is assumed to hold for all $t$ 
in the sense of generalised functions in $x$.

\medskip

\subsection {Gagliardo--Nirenberg estimates}
The result below -- the  1d version of the 
Gagliardo--Nirenberg inequalities --  is of fundamental importance for what follows.
For $1\le p\le\infty$, $m \in\N_0=\N\cup 0$ and a function $h$ on $ S^1$ (not necessarily with zero mean-value) 
we denote
\be\label{norm}
|h|_{m,p}= |\pa^m h|_p +  | h|_p ,
\qquad \text{where} \;\;
 |u|_p:=|u|_{L_p}.
\ee
\bl\la{lGN}
Let $m\in\N$ and  $\beta\in\N_0$,  $\beta\le m-1$. Let 
$q,r\in[1,\infty]$. Then
\smallskip\\
a) if $p\in(0,\infty)$, and $\theta$ found from the relation
\be\la{2.2}
\beta-\fr1r=\theta(m-\fr1p)-(1-\theta)\fr 1q
\ee
satisfies $\theta\in[\fr\beta m,1)$, then
\be\la{2.3}
|h|_{\beta,r}\le C|h|_{m,p}^\theta |h|_{q}^{1-\theta},
\ee
with some $ C=
C(\beta,p,q,r,m)$.
\smallskip\\
b) If $p=1$ or $p=\infty$, then (\ref{2.3}) holds if in addition to (\ref{2.2}) also $\theta=\beta/m$.

\el

\bexs\la{ex1}
{\bf A.} Choosing $r=p=q=\infty$ we get for any  $b, m \in \N_0$, $b\le m-1$, 
 the Landau-Kolmogorov inequality
$
|h|_{C^b}\le C|h|^\theta_{C^m}|h|^{1-\theta}_{\infty}$, where 
$\theta= b/m.
$

\noindent
{\bf B.} If $p=q=2$, $r\ge 2$ and  $0\le k\le m-1$, then
$$
|h|_{k,r}\le C\Vert h\Vert^\theta_m \Vert h\Vert^{1-\theta},
\qquad \theta=\fr{2rk+r-2}{2rm}.
$$
{\bf C.} If $1\le k \le  m-1$ and $2m\ge k+1+(2m-2)/r$, then
applying (\ref{2.3}) to $h_x$, we get:
\be\la{exC}
|h|_{k,r}\le C\Vert h\Vert^\theta_m | h|^{1-\theta}_{1,1},
\qquad \theta=\fr2r\fr{rk-1}{2m-1}.
\ee
\eexs

\subsection{Well-posedness of the deterministic equation}
The following result is obtained  in 
 \cite[Section 1.3]{BK} by a very traditional  application of the   Galerkin  method  \ci{L1969}.
 We sketchy recall the proof. \

\bt\la{t2}
Let $0<T<\infty$, $m\in\N$ and let $m_*$ be  any real number bigger than  $m$.
If $u_0\in H^m$ and $\xi\in {X}_T^{m_*}$, then the initial-value problem (\ref{B2}) has a unique solution  $u\in X_T^m$. 
Moreover, the a priori bound holds,
\be\la{apri}
|u|_{X_T^m}\le C(m,T,\nu, \Vert u_0\Vert_m,|\xi|_{X_T^{m_*}}).
\ee
\et

For $N \in\N$, let  $H^{(N)}$ be the $2N$-dimensional 
subspace of $H$, spanned by the vectors $\{e_s,\,|s|\le N \}$,
so
 $H^{(N)}\subset H^m$ for all $m$. Denote by 
$
\Pi_N: H\to H^{(N)}
$
the orthogonal projection. Then 
$
\Pi_N\Big(\sum_{-\infty}^\infty v_s e_s\Big)=\sum_{-N}^N v_s e_s, 
$
and $\Pi_N$ commutes with the operator $\fr{\pa^2}{\pa x^2}$. 
Let us substitute in (\ref{B2}) the sum
$ 
u=u^N(t)=\sum_{-N}^N u_s^N(t) e_s,
$\ 
and apply to that $\Pi_N$.  We obtain
\be\la{BN}
\pa_t u^N-\nu u^N_{xx}+ \Pi_N(u^N u^N_x) =
\pa_t\Pi_N\xi(t),\qquad u^N(0):=\Pi_N u_0.
\ee
This is the $N$-th Galerkin approximation for problem  \eqref{B2}.
For  $v\in H^{(N)}$  its nonlinear term $\Pi_N(vv_x)$ is $L_2$-orthogonal to $v$:
\be\la{*}
\langle\Pi_N(vv_x),v\rangle
=\langle vv_x,v\rangle=\frac13\int_{S^1}\pa_x v^3=0 .
\ee
  Denoting by $v^N(t)$ a solutions of the linear equation, obtained from   \eqref{BN} by removing the term ${\Pi_N(u^N u^N_x)}$,
and writing   in (\ref{BN})  $u^N=v^N+w^N$, we get for $w^N(t)$ an ODE in $H^{(N)}$ 
with  continuous  in $t$ coefficients.  Using the 
orthogonality \eqref{*} we easily get that  this equation   has a unique solution, defined for all $t\ge0$. 
 So \eqref{BN} also  has a unique solution $u^N$. 
 When $N\to\infty$, the solutions $u^N$ converge to a solution of 
\eqref{B2}:

\bl\la{l4}
For any $T>0$,  $m_*>m\in\N$, $u_0\in H^m$ and $\xi\in{X}_T^{m_*}$ solutions $u^N(t)$ of (\ref{BN})
converge to a unique solution $u(t)$ of  (\ref{B}) as $N\to\infty$, weakly  in   space
$X_T^m$, as well as strongly  in   spaces $X_T^{m-1}$ and  $L_2(0,T;H^m)$.
\el

For a proof see \ci{BK}. This lemma is a useful  tool to study   properties of solutions for 
stochastic PDE \eqref{B2}   with  random force  \eqref{xi}  since it allows to 
approximate this {\it  infinite-dimensional} system by  {\it  finite-dimensional}  stochastic systems \eqref{BN} (for stochastic ODE e.g.  see
\cite{Oks} ).

 \smallskip
Theorem \ref{t2} implies that for any $m_*>m\in\N$ and $0<T<\infty$ we can define the mapping
\be\la{M}
\cM=\cM^{T,\nu}: H^m\times{X}_T^{m_*}\to X_T^m,
\quad
(u_0,\xi)\mapsto u(\cdot)|_{[0,T]},
\ee
and for  $0\le t\le T$ -- the mappings
\be\la{Mt}
\cM_t=\cM_t^{T,\nu}: H^m\times{X}_T^{m_*}\to H^m,
\quad
(u_0,\xi)\mapsto u(t),
\ee
where $u(t,x)$ solves \eqref{B2}. 
Certainly, if $T_1\ge T_2\ge t\ge 0$, then
$$
\cM_t^{T_1,\nu}(u_0,\xi)=\cM_t^{T_2,\nu}(u_0,\xi|_{[0,T_2]}).
$$
So solutions $u(t,x)$ of \eqref{B2} are defined for all $t\ge0$.

 A map of Banach spaces $F :B_1\to B_2$  is called {\it locally bounded } if for any  $ R>0$, 
$F(\ov B_{B_1}^R)$ is a bounded set in $B_2$. 
It is called {\it locally Lipschitz},  if for  any $R>0$ its restriction  
$F|_{\ov B_{B_1}^R}$ is a Lipschitz mapping.

\bt\la{t3}
Under the assumptions of Theorem  \ref{t2}  the mappings 
$\cM^{T,\nu}$ and $\cM^{T,\nu}_t$ with $0\le t\le T$ are
 locally Lipschitz, for any $m\in\N$.
 In particular, they are locally bounded  and  continuous.
\et
See \cite[Section 1.3]{BK}

Solutions, constructed in Theorem \ref{t2}, possess an  important non-expanding property, needed below:
  
  \bl\la{lnexp}
   For a fixed    $ \xi \in  {X}^{2}_\infty$     and $u_1, u_2 \in H^1$ 
  let $u^j(t,x)=:u^j(t)$ solves eq.~\eqref{B2} with $u_0 =u_j$, $j=1,2$.  Then for  any $T\ge0$, 
   \be\la{nexp}
  \big| u^1(T) - u^2(T)\big|_1  \le |u_1-u_2|_1.
  \ee
  \el
  \bpr
  Let first $u_1, u_2 \in H^\infty$. 
  Denoting $w=u^1-u^2$ we see that $w$   satisfies 
  \be\label{h1}
  w_t + \tfrac12 (w(u^1+u^2))_x -\nu w_{xx} =0,\quad w(0,x) = u_1-u_2 =: w_0. 
  \ee
  Let us consider the conjugated Cauchy problem 
  \be\label{h2}
  {\phi}_t + \tfrac12 {\phi}_x(u^1+u^2)  + \nu {\phi}_{xx} =0,\; \; 0\le t\le T; 
  \quad {\phi}(T,x) = {\phi}_T(x),  
  \ee
where $\phi_T$ is a Lipschitz function such that $| \phi_T |_\infty =1$. For any $\phi_T$ like that  problem \eqref{h2} has a unique
classical solution $\phi$.  The maximum principle applies to the equation in \eqref{h2}, and so
$
| \phi(t)|_\infty \le1
$
for 	all $t\in[0,T]$. Multiplying  the equation in \eqref{h1} by $\phi$ and integrating by parts we get that 
$$
|\langle w(T), \phi_T \rangle |=| \langle w_0, \phi(0)\rangle| \le |w_0|_{L_1}, 
$$
for any $\phi_T$ as above.  Now let $\chi_\eps(w)$, $0<\eps \le1$, be a sequence of piece-wise linear  continuous 
 functions on $\R$,  for each  $w$  converging to sgn$\,(w)$
  as $\eps\to0$, and such that $|\chi_\eps|_\infty \le1$.
 Since by Theorem~\ref{t2}  functions $u^1,u^2$ are smooth in $x$, then we can take for $\phi_T$ any function 
 $\chi_\eps(w(T,x))$. Thus, 
 $\ds \int w(T,x) \chi_\eps(w(T,x))dx \le   |w_0|_{L_1}$ forall $\eps. $
 Passing to the limit as $\eps\to0$ we get that 
   $
   | w(T)|_1 =
 \big| u^1(T) - u^2(T)\big|_1  \le |u_1-u_2|_1.
  $
  By continuity the estimate stays true for $u^1,u^2 \in H^1$. 
   \epr

\setcounter{equation}{0}
\section {Stochastic initial-value problem}\label{s_4}

Everywhere below, if in the stochastic initial-value problem \eqref{B2} the initial data $u_0$ is a r.v. (random variable), it is assumed
to be independent of  the random force \eqref{xi}; e.g. $u_0$ may be a  non-random function.
Let in \eqref{B2} $\xi$ be the random field \eqref{xi},  let 
$0\le T\le \infty$, $m\in\N$ and let $u_0\in H^m$ be a r.v.
(independent of   $\xi$). 
\bd
A random process $u^\om(t)\in H^m$ is a strong solution in $H^m$
of the Cauchy problem (\ref{B2}) for $0\le t\le T$ 
if there exist  a null-set
$Q\subset\Om$ such that for any $\om\in\Om\setminus Q$
the curve  $u(t):=u^\om(t,\cdot)$ satisfies \eqref{dB} for $0\le t\le T$.\footnote{ If $T=\infty$, the interval $[0,T]$ becomes $[0, \infty)$, 
and \eqref{dB} holds for $t<\infty$. }
\ed

The result below is an obvious consequence from Theorem \ref{t2} (we recall that the force $\xi$ is assumed to be as in Proposition~\ref{p1}). 

\bt\la{t5}
Let $m\in\N$ and 
 $u_0\in H^m$ be a r.v., independent of  $\xi$. Then for any
$T\in[0,\infty)$ the Cauchy problem (\ref{B2}) admits a  strong solution
\be\la{strs}
u^\om(t;u_0)= \cM_t^{ T,\nu}
(u_0^\om,\xi^\om|_{[0,T]}), \quad 0 \le t\le T,
\ee
satisfying \eqref{B2} for all 
$\om\in\Om\setminus Q$, where $Q$  is the null set as in Proposition~\ref{p1}. 
For each $\omega$,   $u^\om(\cdot)$ belongs to $X^m_T$, and for every 
$\om\in\Om\setminus Q_{m_*}$ it is 
a limit as in 
 Lemma~\ref{l4} of solutions $u^{N,\om}(t)$ for stochastic differential equations  (\ref{BN})
with $\xi=\xi^\om$ and $u_0=u_0^\om$. Solution \eqref{strs}
  is unique in the sense that any other strong solution coincides with it, a.s. 
\et

The solutions, constructed in Theorem \ref{t5},  are random  processes $u^\om(t) \in H^m$,  $t\ge0$.

We will denote the strong solution of (\ref{B2}), built in Theorem 
\ref{t5}, as
$$
u(t;u_0,\xi)=u(t;u_0)=\cM_t(u_0,\xi).
$$

Next  we obtain basic  a priori estimates for  moments of solutions $u(t;u_0)$. 
We will only sketch their derivations, formally applying  Ito's formula. Rigorously the estimates may be proved by applying the formula to  
the  Galerkin approximations and then using Lemma~\ref{l4} to pass to a limit,  similarly to the proof of Theorem \ref{t2}, sketched   above. 

Let us denote by $f$ the mapping 
 $f(u)=\nu u_{xx}- uu_x
$
and  write the Cauchy problem (\ref{B2}), \eqref{xi}  as
\be\la{BN2}
\dot u(t)=f(u(t))+\sum_s b_s\pa_t\beta_s^\om(t)e_s,
\qquad u(0)= u_0.
\ee 
Let $u(t)$ be a solutions for \eqref{BN2}. 
For a $C^2$-smooth functional $F(u)$ of $u\in H$  the  formal\,\footnote{``Formal" since we do not discuss the properties of the 
solution $u$ and requirements on $F$, needed for the formula to hold.} 
weak 
 Ito formula  (obtained by taking expectation of the usual Ito formula, e.g. see  in \ci{Oks})), reeds
\be\la{Ito}
\fr\pa{\pa t}\EE F(u(t),t)= \EE
\langle\na F(u,t),f(u(t)) 
\rangle+\fr12\EE\sum_s b_s^2\fr{\pa^2}{\pa u_s^2}F(u(t)), 
\ee
in the sense of distributions in $t$ (that is,  (\ref{Ito})   is equivalent  to its integrated in time version).

\subsection{Exponential  $L_2$  moments}

Ito's formula (\ref{Ito}) allows  to get a priori bounds for  solutions of the stochastic equation   (\ref{B2}),  \eqref{xi}. 
Choosing there $F(u)=e^{\sigma'\Vert u\Vert^2}$, where $\sigma'>0$ depends on $\nu $ and $B_0$, we formally derive  from (\ref{Ito}) that 
\be\la{2dies}
\fr \pa{\pa t}\aE e^{\sigma'\Vert u(t)\Vert^2}\le 
C\Big(
\aE e^{\sigma'\Vert u(t)\Vert^2}(-  G \Vert u(t)\Vert^2+B_0)
\Big), \quad t\ge0, 
\ee
where $C$ and $G $ depend on $\nu$ and  $B_0$.
Considering separately the case when 
$- G \Vert u(t)\Vert^2+B_0\ge -1$ and $- G \Vert u(t)\Vert^2+B_0\le -1$, we find that the r.h.s. of 
(\ref{2dies}) is less than 
$-\aE e^{\sigma'\Vert u(t)\Vert^2} + C$. So by the 
Gronwall inequality,
\be\la{Gr}
\aE e^{\sigma'\Vert u(t)\Vert^2}\le 
e^{-t}
\aE e^{\sigma'\Vert u_0\Vert^2}+ C',\qquad t\ge0,
\ee
where $\sigma' =\sigma'(\nu, B_0)$ and  $C'= C' (\nu, B_0)$. We repeat that  to derive this estimate rigorously 
we should apply Ito's formula to solutions of the Galerkin approximations \eqref{BN}
and then pass to a limit as $N\to\infty$, using  Lemma~\ref{l4}. See in \cite{BK}.

\subsection {Moments of higher Sobolev norms }

Applying  (\ref{Ito}) to 
$F(u)=\Vert u\Vert_m^2$ with $m\in\N$, we get that
\be\la{2.4}
\fr\pa{\pa t}\aE \Vert u(t)\Vert_m^2=
-2\nu\aE \Vert u(t)\Vert_{m+1}^2
-\aE\langle u,\pa_x(u_x)^2\rangle_m+B_m, 
\ee
where $\langle\cdot,\cdot\rangle_m$  is the scalar product in
$H^m$. 
From the Gagliardo--Nirenberg inequality (Example~\ref{ex1}.{\bf B}) it follows the estimate
$$
|\langle u,\pa_x u^2\rangle_m|\le C
\Vert u\Vert^{\fr{2m+3}{2m+2}}
\Vert u\Vert^{\fr{4m+3}{2m+2}}_{m+1}, \quad m\in \N.
$$
Applying then   Young's inequality  we get that 
\be\la{Yin}
|\langle u,\pa_x u^2\rangle_m|
\le  \nu \Vert u\Vert_{m+1}^2+C_m(\nu)
\Vert u \Vert^M \le \nu \Vert u\Vert_{m+1}^2+  C(M,\sigma', \nu) e^{\sigma' \| u \|},
\ee
for some $M=M(m)$.
Now from (\ref{2.4}),  (\ref{Yin}), (\ref{Gr}) and Gronwall's
inequality  it follows that
\be\la{2.5}
\aE \Vert u(t)\Vert_m^2\le C_m
\Big[
1+\aE \Vert u_0\Vert_m^2+\aE e^{\sigma'\Vert u_0\Vert^2}
\Big],
\ee
and then -- that 
\be
\nu\aE \int_0^t\Vert u(s)\Vert_{m+1}^2ds\le C_m
\Big[
\aE \Vert u_0\Vert_m^2+\aE e^{\sigma'\Vert u_0\Vert^2}
\Big],\!\!\qquad \qquad\qquad\la{2.6}
\ee
where $C_m$ depends on $m,\nu,\sigma'$, and $B_m$.

\bt\la{t2.3}
If  for some  $m\in\N$  $u_0$ is a r.v. in $H^m$, 
 independent of  $\xi$,   then solution $u(t;u_0)$ satisfies 
(\ref{Gr}), (\ref{2.5}) and (\ref{2.6}) for all $t\ge0$, where the constants depend on $\nu$.
\et
 If $m=0$, then by \eqref{*} the second term in the r.h.s. of 
(\ref{2.4}) vanishes. So, formally,
\be\la{patV}
\fr\pa{\pa t}\aE \Vert u(t)\Vert^2=-2\nu
\aE\Vert u(t)\Vert_1^2+B_0.
\ee
This equality can be rigorously justified,  using Lemma \ref{l4} and 
Theorem~\ref{t2.3} with $m=2$.

\setcounter{equation}{0}
\section{The Markovness}\label{s_5}

\subsection{The law of a solution} 
Let $u_0$ be a r.v. in $H^m$, $m \in \N$, independent of  $\xi$.  Then  a solution $u^\om(t):=u(t;u_0)$  defines a r.v.
$
\om\mapsto u^\om(\cdot)\in X_T^m.
$
Its law 
$\mu:=\cD(u)\in \cP(X_T^m)$ is a (Borel) measure on $X_T^m$,  and  for any function $f\in C_b(X_T^m)$,
$
\ds\int_{X_T^m} f(u)\mu(du)=\aE f(u^\om(\cdot)).
$
Consider the r.v.
$$
\Psi:\om\mapsto(u_0^\om,\xi^\om(\cdot))\in H^m\times{X}_T^{m_*}.
$$
Its law is a measure $\cD(\Psi)\in\cP(H^m\times{X}_T^{m_*})$.
Since $u_0$ and $\xi$ are independent, then
$
\cD(\Psi)=\cD(u_0)\times\cD(\xi(\cdot)),
$
where $\cD(u_0)\in\cP(H^m)$ and 
$\cD \xi(\cdot)\in\cP(X_T^{m_*})$. So, finally,
\be\la{2.7}
\cD(u(\cdot;u_0))=\cM \circ \big(\cD(u_0)\times \cD(\xi(\cdot))\big),
\ee
where $\cM \circ  \lam$ stands for the image of a measure $\lam$ under the mapping $\cM$ (see \eqref{M}).\footnote{That is, 
$
(\cM\circ \lam)(Q)=\lam(\cM^{-1}(Q))$ for $ Q\in\cB(X_T^m).
$}
Similarly,  for $0\le t\le T$, 
$$
\cD(u(t;u_0))=\cM_t \circ \big(\cD(u_0)\times \cD(\xi(\cdot))\big).
$$
Due to \eqref{2.7} the distribution $\cD\big( u(\cdot; u_0)\big)$ will not
change if in  series \eqref{xi} we replace $\{\beta_s^\om(t)\}$ by another set of 
standard independent Wiener processes.

Note that in the case when the initial function $u_0$ is non-random,  its law is the  delta-measure
$\cD(u_0)=\de_{u_0}$.

\bd
The transition probability for the Burgers equation (\ref{B}) is the mapping
$$
\Si:\R_+\times H^m\to\cP(H^m), \quad 
(t,u_0)\mapsto\Si_t(u_0)=\cM\circ(\de_{u_0}
\times\cD(\xi(\cdot))).
$$
For $Q\in\cB(H^m)$ we will write
$
\Si_t(u_0)(Q)=:\Si_t(u_0,Q).
$
\ed

Then for a (non-random) $u_0\in H^m$
we have
$
\Si_t(u_0,Q)=\aP(u(t;u_0)\in Q),
$
and for $f\in C_b(H^m)$
$
\int_{H^m}f(v)\Si_t(u_0,dv)=\aE\big(f\big(u(t;u_0)\big)\big).
$

 For $0 <T\le\infty$ we  denote $$W_T^{m_*} :=\cD(\xi \mid\!_{[0,T]} )\in \cP({X}_T^{m_*}).$$
 This is a Wiener measure on ${X}_T^{m_*}$ (the latter  is a Banach space  for
  $T<\infty$ and is a complete 
separable metric space  for  $T=\infty$).

\subsection {The semigroup in mesures}\label{ss_6.2}
For $t\ge 0$ denote
$$
S_t^*:\cP(H^m) \to \cP(H^m), \quad 
\mu\mapsto S_t^*(\mu) = \cD\big(u(t;u_0)\big),
$$
where $\cD(u_0)=\mu$. Then for any $f\in C_b(H^m)$
and a r.v. $u_0$ with a law $\cD(u_0)= \mu$, we have   that 
$$
\langle f,  S_t^*(\mu)\rangle=\aE f(u(t;u_0))=
\int_{{X}_T^{m_*} \times H^m } f\big(\cM_t(u_0,\xi)\big)
dW_T^{m_*}(\xi)d\mu(u_0).
$$
Applying  Fubini's  theorem, we get
$$
\langle f,  S_t^*(\mu)\rangle=
\int d\mu(u_0)\Big(\int f\big(\cM_t(u_0,\xi)\big)dW_T^{m_*}(\xi)\Big)
=
\int d\mu(u_0)\Big(\int f(v)  \Si_t(u_0,dv)\Big).
$$
It shows that 
\be\la{2.8}
S_t^*(\mu)=\int\Si_t(u_0)d\mu(u_0).
\ee
The r.h.s. is   a measure on $M$ such  that for  any $f\in C_b(H^m)$,
\be\label{from}
\langle f,  \int\Si_t(u_0)d\mu(u_0) \rangle= \int \langle  f,\Sigma_t(u_0)  \rangle d\mu(u_0).
\ee
In the next subsection  we show that the transformations $S_t^*$ form a semigroup. 

\smallskip

A measure $\mu_\nu \in \cP(H^m)$ is called {\it stationary} 
for  stochastic equation \eqref{B2}, \eqref{xi} if $S_t^* \mu_\nu=\mu_\nu$ 
for all $t\ge0$. In view of Theorem~\ref{t2.3} the well known Krylov--Bogolyubov argument applies to the equation and easily implies 
that for any $m\in\N$
\be\label{KB}
\text{a stationary  measure \ $\mu_\nu \in \cP(H^m)$ \ exists}.
\ee
Moreover, due to \eqref{B_assume}  there exists a stationary 
 measure $\mu_\nu$ such that 
$\mu_\nu(H^\infty) =1$.

\subsection {The Chapman--Kolmogorov  relation and   Markovness
}

Let the initial state $u_0\in H^m$ be non-random.
For  $t_1,t_2\ge 0$ let us  consider  the r.v. 
$u(t_1+t_2; u_0) \in H^m$ and its law 
$\Si_{t_1+t_2}(u_0)\in\cP(H^m)$.
Denote $v(t):=u(t_1+t;u_0)$ and $\zeta^\om(t):=\xi^\om(t_1+t)-\xi^\om(t_1)$. Then $v(t_2)=u(t_1+t_2;u_0)$, and
$$
\pa_t v+vv_x-\nu v_{xx}=\pa_t\xi^\om(t_1+t)
=\pa_t\zeta^\om(t),
\quad t\ge 0;\qquad v(0)= u^\om(t_1; u_0).
$$
We have 
$
\cD(v(0))=\Si_{t_1}\big(u_0\big)=:\mu_1.
$
By (\ref{xi}), 
$
\zeta^\om(t)=\sum b_s\big( \beta_s^\om(t_1+t)-\beta_s^\om(t_1)\big) e_s(x).
$
But the set of random processes $\{\beta_s^\om(t_1+t)-\beta_s^\om(t_1),\; t\ge 0\}$ is another collection  of 
 standard independent Wiener processes. So by (\ref{2.8}),
\begin{equation*}
\ba{ccc}
\cD\big(v(t_2)\big)&=&\ds\int \Si_{t_2}(v_0)d\mu_1(v_0)
\\
\Vert&&\Vert
\\
\cD\big(u(t_1+t_2;u_0)\big)=\Si_{t_1+t_2}(u_0)&&\ds\int \Si_{t_2}(v_0)\Si_{t_1}(u_0;dv_0).
\ea
\end{equation*}
Thus, we have proved the Chapman--Kolmogorov relation
\be\la{K-Ch}
\Si_{t_1+t_2}(u_0)=\int \Si_{t_2}(v)\Si_{t_1}(u_0;dv).
\ee

Now consider $0\le t_1\le t_2\le t_3$ and denote $\De=[0,t_3-t_2]$.
 Take any $F\in C_b(X^m_{t_3-t_2})$. We are going to calculate $\aE F\big( u(t_2+\tau\,;u_0)\big)$ with $\tau\in\De$. Here  $\tau \mapsto u(t_2+\tau\,;u_0)$ is the trajectory, starting from $u_0$
 at $t=0$
and restricted to $[t_2,t_3]$, which is a shifted interval $\De$.
Arguing exactly as when proving (\ref{K-Ch}), 
we get 
\be\la{3.4}
\aE F\big( u(t_2+\tau\,;u_0)\big)=\int_{H^m}
\aE F\big( u(t_2-t_1+\tau\,;v)\big)\Si_{t_1}(u_0;dv),
\ee
where in the r.h.s.  $u(t_2-t_1+\tau\,;v)$
is the trajectory, starting from $v$  at $\tau=0$
and restricted to $[t_2-t_1,t_3-t_1]$ (which is also a shifted interval $\De$).
\medskip

The Chapman--Kolmogorov relation (\ref{K-Ch}) implies 
that the operators $\{S_t^*:t\ge 0\}$ form a semigroup:
\be\la{3.5}
S_{t_1+t_2}^*=
S_{t_1}^*\circ S_{t_2}^*,\qquad t_1,t_2\ge 0; \qquad S_0^*= \text{id}.
\ee
Indeed, for any $\mu\in\cP(H^m)$  let us integrate  (\ref{K-Ch})
against the measure $\mu(du_0)$. Then by (\ref{2.8}) the l.h.s. is
$ \,
\int_{H^m}\Si_{t_1+t_2}(u_0)\mu(du_0)=S_{t_1+t_2}^*(\mu), \,
$
while the r.h.s. is
$$
\int_{H^m}\Si_{t_1}(v)\int_{H^m} \Si_{t_2}(u_0;dv)\mu(du_0)
=\int_{H^m}\Si_{t_1}(v) [ S_{t_2}^*(\mu) (dv) ]
=S_{t_1}^*\circ S_{t_2}^*(\mu),
$$
so  (\ref{3.5}) is proved. 

This  implies  that solutions $\{u^\om(\tau;u_0): u_0\in H^m\}$ make  a family of Markov
processes in $H^m$ and the semigroup $\{S_t^*:t\ge 0\}$
transforms their distributions.

Apart from transformations $S_t^*$ of the space of measures on $H^m$, solutions $u(t;u_0)$ define linear transformations $S_t$, $t\ge0$, of the space $C_b(H^m)$:
$$
S_t:  C_b(H^m) \to C_b(H^m), \quad S_t(f)(v) = \EE f(u(t;v)) =\lan f, \Sigma_t (v)\ran.
$$
(The function $v\mapsto \EE f(u(t;v))$  obviously is bounded, and it 
 follows from the continuity of the mappings $\cM_t$ that it is continuous). 
From \eqref{from} it follows that the transformations $S^*_t$ and $S_t$ obey the duality relation:
\be\label{dual}
\langle S_t f, \mu\rangle = \langle f, S_t^*\mu\rangle. 
\ee
It is  easily seen from \eqref{3.5} and \eqref{dual} that the transformations  $S_t$ also  form a semigroup. 

The Markovness of the random process, defined by solutions of the stochastic Burgers equation, 
 is as well a   characteristic feature  for  other  well-posed
stochastic PDEs, e.g. for the stochastic 2d Navier--Stokes
equations; see \ci{DZ, Oks,KS}. Now we move to results which are specific for 
the Burgers equation (\ref{B}).

\setcounter{equation}{0}
\section {Improving  upper estimates via Oleinik maximum principle.}\label{s_olei}

Let $u(t,x)$ solves the free Burgers equation:
\be\la{fB}
u_t(t,x)+u u_x-\nu u_{xx}=0,\qquad t\ge 0;
\quad
u(0,x)=u_0(x), 
\ee
and assume first that $u_0\in C^\infty (S^1)$. 
Then $u(t,x)$ is a smooth function. 
   Differentiating \eqref{fB} in $x$, multiplying by $t$  and denoting  $w=tu_x$ 
we get:
\be\la{(*)}
(w_t-u_x)+tu_x^2 +uw_x=\nu w_{xx},\qquad w(0,x)=0.
\ee
Consider the function $w(t,x)$ on the cylinder $Q=[0,T]\times S^1$ and denote
$$
M=\max w|_Q.
$$
since $\ds\int w(t,x)dx=0$, then $M\ge 0$.
If $M=0$ then  $u(t,x)\equiv 0$. This   is a trivial case.
Now let $M>0$.  Then the  maximum $M$
is attained at a point $(t_1,x_1)$ with $t_1>0$. By the criterion of maximum, 
\be\la{(**)}
w_t\ge 0, \quad w_x=0,\quad w_{xx}\le 0 \qquad{\rm at}\,\,\,(t_1,x_1).
\ee
From (\ref{(*)}) and  (\ref{(**)}) we obtain that at  $(t_1,x_1)$ we have
$\ 
-u_x+tu_x^2=\nu w_{xx}-w_t\le 0.
$
Multiplying by $t$ and using that $tu_x=w$  we get that 
$
-w+w^2\le 0$ at $ (t_1,x_1).
$ 
Hence,
$
-M+M^2\le 0$, or $M(M-1)\le 0.$
Since $M > 0$, then $M\le 1$. Thus, we got
\smallskip\\
{\bf $1^\circ$.} If  $u_0\in C^\infty\cap H$, then
\be\la{(dies)}
t u_x(t,x)\le 1,\qquad \forall\, t\ge 0,\,\,\,\forall\, x\in S^1.
\ee
Approximating $u_0\in H^1$ by smooth functions, we get that
\smallskip\\
{\bf $2^\circ$.} If  $u_0\in H^1$, then 
(\ref{(dies)}) still holds.

As usual,  we write a function $v(x)$ as 
$
v=v^+-v^-,
$
where  $v^+(x)=\max(0,v(x))$ and $v^-(x)=\max(0,-v(x)).$
\bl 
\la{lemma}
Let $v\in C^1(S^1)$, $\ds\int vdx=0$, and $v_x(x)\le C_*$, $C_*\ge0$. 
Then 
\be\la{lab} 
|v|_\infty\le C_*,\qquad |v_x|_1\le 2C_*.
\ee
\el
\bpr
Since $v(x)$ is periodic with zero mean, then $v(x_0)=0$
for some $x_0\in S^1$. Therefore, $v(x)=\ds\int_{x_0}^x v_y(y)dy\le C_*$ for $x\in [x_0,x_0+1]$.
Similarly, $-v(x)=\ds\int_x^{x_0} v_y(y)dy\le C_*$ for $x\in [x_0-1,x_0]$. Thus, the first inequality of (\ref{lab}) is proved.
For the second inequality, note that $\ds\int (v_x)^+dx\le C_*$. Since 
$\ds\int_{S^1} v_x(x)dx=0$, then  $\ds\int_{S^1} (v_x)^-dx=\ds\int_{S^1} (v_x)^+dx$. This  implies the second inequality as 
$ |v_x| ={v_x}^+ +{v_x}^-$. 
\epr

The result below was proved by O.A. Oleinik \cite[Section~7]{Olei} for solutions of problem \eqref{B2} with a regular force $\eta$. Its proof 
in \cite{BK} which we now discuss follows the argument in \cite{Olei}. 
\bt\la{t3.1}
Let $u(t,x)$ solves non-random problem (\ref{B2}) with
$\xi\in {X}_T^4$ and $u_0\in H^1$. Then for  all $\nu$ and any $T>0$, $\theta\in (0,T]$, we have 

i) $\sup_{\theta\le t\le T}|u_x^+(t,\cdot)|_\infty
\le 4\theta^{-1}(1+T|\xi|^{1/2}_{X_T^4}+T|\xi|_{X_T^4})$, 

ii) $\sup_{\theta\le t\le T}|u(t,\cdot)|_\infty
\le 4\theta^{-1}(1+T|\xi|^{1/2}_{X_T^4}+T|\xi|_{X_T^4})$,

iii) $\sup_{\theta\le t\le T}|u_x(t,\cdot)|_1
\le 8\theta^{-1}(1+T|\xi|^{1/2}_{X_T^4}+T|\xi|_{X_T^4}).$
\et

\bpr
 It suffices to prove i) since ii) and iii) follow from it by the lemma
above.  If $\xi=0$, then i) is already  proved. 
 If $\xi\ne 0$ and $u_0\in C^\infty$, write $u=\xi+v$. Then
$\ 
v_t+uu_x-\nu v_{xx}=\nu\xi_{xx}.
$
Denote $w=tv_x$ and consider $w$ as a function on 
the cylinder $Q=[0,T]\times S^1$ with $M=\max w|_Q$.
Then again the criterion of maximum  (\ref{(**)}) holds at a point
$(t_1,x_1)$ where $M$ is achieved. Arguing as above when proving (\ref{(dies)}), we get i). 

For $u\in H^1$, we approximate it in $H^1$ by $u_0\in C^\infty$ and pass to a limit. See  \ci{BK} for details.
\epr

Now let us pass to the stochastic Burgers equaution. 

\bt\la{t3.2}
If $u_0\in H^1$ is a r.v. (independent of  $\xi$) and $u^\om(t)=u^\om(t;u_0)$, then for any 
$ p\in[1,\infty)$,  $0<\theta\le 1$ and $ T\ge \theta$
we have:

i) $\aE\sup\limits_{T\le t\le T+1}|u_x^+(t)|^p_\infty \le C\theta^{-p}$, 

ii) $\aE\sup\limits_{T\le t\le T+1}
\big[|u(t)|^p_\infty+|u_x(t)|^p_1\big]\le C\theta^{-p}$, 
\noindent 
where $C=C(p, B_4)$. 
\et
\bpr
By Lemma \ref{lemma}, it suffices to prove i).
For any $v(t,x)$ denote 
$
\Phi(v)=\sup\limits_{\theta\le t\le\theta+1}
|(v_x)^+(t,\cdot)|_\infty^p.
$
Further we consider the cases $T=\theta$ and $T>\theta$
separately.
\smallskip\\
{\bf The case $T=\theta$.} By i) in Theorem \ref{t3.1},
$
\Phi(u^\om)\le C_p\theta^{-p}(1+|\xi^\om|^p_{X^4_{\theta+1}}) \ 
$
for each $\om$. 
So by Proposition~\ref{p1}, 
\be\la{est}
\aE \Phi(u^\om)\le C_p(B_4) 
\theta^{-p}. 
\ee
This proves i) with $T=\theta$ for a non-random $u_0$. 
If $u_0$ is random, then we integrate   estimate  (\ref{est}) over $\cD(u_0)$ and again get i). 
 \smallskip\\
{\bf The case $T>\theta$.} Now $T=\vka+\theta$ where 
$\vka>0$. By the Chapman--Kolmogorov relation  in the form (\ref{3.4}),
the l.h.s. of i) with a non-random $u_0$ is
$$
\int_{H^1}\Si_\vka(u_0;dv)\,\aE\Phi\big(u(t;v)\big)
\le C\theta^{-p}, 
$$
where we used (\ref{est}).
If $u_0$ is random, then integrating over $\cD(u_0)$  we complete the proof.
\epr 

This is the first main upper bound. 
Now we pass to the second.

\bt\la{t3.3}
Let $m\in\N$, $\theta>0$ and $u_0\in H^1$ be non-random. Then
there exists a constant $C_m^*(\theta, B_{\max(4,m)} )$
such that 
\be\la{*3.3}
\aE\Vert u(t;u_0)\Vert_m^2\le C_m^*\nu^{-(2m-1)}
=:Y,\qquad t\ge \theta.
\ee
Estimate \eqref{*3.3}  remains true if $u_0\in H^1$ is a r.v., independent of  $\xi$. 

 \et

\bpr
Consider first the smooth 
case when $u_0\in C^\infty\cap H$.  Denote $u(t;u_0) =: u(t)$. 
Then by \eqref{2.4} 
\be\la{3.6}
 \aE \Vert u(t)\Vert_m^2- \aE \Vert u_0\Vert_m^2
 =\int_0^t\big(
-2\nu\aE \Vert u(s)\Vert_{m+1}^2
-\aE\langle u(s),\pa_x(u_x(s))^2\rangle_m+B_m\big) ds.
\ee
By the Gagliardo--Nirenberg inequality (\ref{exC})  we easily get that
\be\la{3.7}
|\langle  u, \pa_x u^2\rangle_m|\le C_m |u|_{1,1}^{\fr{2m+3}{2m+1}} \Vert  u \Vert_{m+1}^{\fr{4m}{2m+1}},\qquad |u|_{1,1}=|u_x|_1.
\ee
Let $t\ge\theta$ and $u=u(t)$. Then by Theorem~\ref{t3.2}.ii) and  H\"older's inequality,
\be\la{3.8}
\aE \Big(|u|_{1,1}^{\fr{2m+3}{2m+1}} \Vert  u \Vert_{m+1}^{\fr{4m}{2m+1}}\Big)
\le
\Big(\aE|u|_{1,1}^{2m+3}\Big)^{\fr1{2m+1}} 
\Big(\aE\Vert  u \Vert_{m+1}^2\Big)^{\fr{2m}{2m+1}}
  \le 
C_m'
\Big(\aE\Vert  u \Vert_{m+1}^2\Big)^{\fr{2m}{2m+1}}.
\ee
Denote $X_m(t)=\aE\Vert u(t)\Vert_m^2$ for $m\in\N$.
Then by (\ref{3.6}) and (\ref{3.7})--(\ref{3.8}) 
\be\la{3.9}
\fr d{dt} X_m\le B_m-2\nu X_{m+1}  +
 C_m X_{m+1}^{\fr{2m}{2m+1}}.
\ee
Using again  (\ref{exC}),    H\"older's inequality
 and Theorem~\ref{t3.2}.ii)  we obtain  
$$
X_m(t) \!=\!\aE \| u(t)\|_m^2 \!\le\! C_m  \aE \| u(t)\|_{m+1}^{2\kappa_m}   | u(t)|_{1,1}^{2(1\!-\!\kappa_m)} 
\!\le \!C_m \Big(  \aE \| u(t)\|_{m\!+\!1}^{2}\Big)^{\!\kappa_m}
\Big(  \aE | u(t)|_{1,1}^{2}\Big)^{\!1\!-\!\kappa_m}\!, \,\,\, \kappa_m\!=\!\tfrac{2m\!-\!1}{2m\!+\!1}.
$$
So in view of Theorem \ref{t3.2}.ii), 
$$
X_m(t)\le C_m(\theta, B_{\max(4,m)} )  X_{m+1}^{\kappa_m}(t), 
\qquad t \ge \theta.
$$
  Thus, 
\be\la{3.10}
X_{m+1}(t)\ge C_m^{''} X_m^{\fr{2m+1}{2m-1}}(t),
\qquad t\ge\theta.
\ee
Let us consider the cases $X_m(\theta)=\aE\Vert u(\theta)
\Vert_m^2 <Y$ and $X_m(\theta)\ge Y$ separately.
\smallskip\\
{\bf Case  $X_m(\theta) <Y$.} If (\ref{*3.3}) is wrong, find the first moment of time $\tau>\theta$ when $X_m(\tau)=Y$. Then 
$
\fr d{dt} X_m(\tau)\ge 0$ and $ X_m(\tau)=Y.
$
Plugging this into (\ref{3.9}) 
 with $t=\tau$ we get that 
$$
0 \le B_m-  X_{m+1}^{\fr{2m}{2m+1}} (\tau)\big( 2\nu   X_{m+1}^{\fr{1}{2m+1}} (\tau) -C_m\big).
$$
 But by \eqref{3.10} we have
$
X_{m+1} \ge C_m{''} (C_m^*)^{\frac{2m+1}{2m-1}} \nu^{-(2m+1)}
$
since $X_m(\tau) = Y$. Therefore 
$$
B_m \ge  C_m^{(1)} ( C_m^*)^{\frac{2m}{2m-1}} \nu^{-2m} \big(  C_m^{(2)} \big( C_m^*)^{\frac1{2m-1}} -C_m\big), 
$$
where the constants $ C_m,  C_m^{(1)}$ and $ C_m^{(2)}$ depend on $\theta$ and $B_{\max(4,m)} $, but not 
 on $C_m^*$.  Since $\nu\le1$, then choosing $C_m^*$ in
\eqref{*3.3} sufficiently big we get a contradiction, which proves \eqref{*3.3}. 
\smallskip\\
{\bf Case  $X_m(\theta)\ge Y$.} The proof in this case is similar, but a bit more involved; see \ci[Section~2.2]{BK}.
There we show that now (\ref{*3.3}) still  holds
for $t\ge 3\theta$, for any positive 
 $\theta$. Re-denoting $3\theta =: \theta$ we get the assertion, if $u_0$ is smooth in $x$.

Finally, for $u_0\in H^1$   we approximate it by  the Galerkin projections  $\Pi_N u_0$. The latter are smooth functions, so for them the 
estimate is already proved. Passing  to a limit as $N\to\infty$ using  Theorem~\ref{t3} and 
 Fatou's lemma we recover \eqref{*3.3} for $u_0$. 
 
 The validity of the last assertion follows  by integrating (\ref{*3.3}) against  the measure $\cD(u_0)$. 
  \epr

 In particular, we see that for any 
 $u_0\in H^1$   solution $u(t;u_0)$ is such that for each  $\theta>0$, 
$
u(\theta; u_0) \in H^\infty = C^\infty \cap H$,  a.s.

\setcounter{equation}{0}
\section {Lower bounds}\label{s_8}
We recall that the rate of energy dissipation $\epsilon^K$ for a flow of real fluid is defined in \eqref{Krate}  (see  more in \cite{LL, F, Fal}). 
Accordingly, for a flow $u^\nu(t,x)$ of ``Burgers fluid" it should be defined as 
$$
\epsilon^B = \nu  \EE\int | u_x^\nu(t,x)|^2dx.
$$
Upper bound for $\epsilon^B$  follows from \eqref{*3.3} with $m=1$. 
Our first goal in this section is to estimate this quantity from below. To do that let  us integrate in time  the balance of energy relation  \eqref{patV} from $T$ to $T+\sigma$: 
 \be\label{en_bal}
 {\mathbb E} \int \tfrac12 |u(T+\sigma, x)|^2dx -  {\mathbb E} \int\tfrac12 |u(T, x)|^2dx  + \nu {\mathbb E} \int_T^{T+\sigma} \!\!\!\int |u_x(s,x)|^2dxds = 
\tfrac12 \sigma B_0,
\ee
where $T, \sigma >0$. 
Let $T\ge1$. Then by  Theorem \ref{t3.1}  the first two terms are bounded by a constant $C_*$, which depends 
only on the random force. If $ \sigma \ge \sigma_* = 4C_*/B_0$, then  $C_*\le \tfrac14 \sigma B_0$ and   we get from \eqref{en_bal} 
 a lower bound  for the rate of  energy dissipation
 locally averaged in time:
 \be\la{en_diss}
   \nu {\mathbb E}\, \frac1{\sigma}  \int_T^{T+\sigma} \!\!\!\int |u^\nu_x(s,x)|^2dxds \ge \tfrac14 B_0.
\ee

\noindent {\it 
 Notation.}
For a  random process $f^\om(t)$ we  denote by the brackets 
 $\llangle f \rrangle $ its averaging  in ensemble and local  averaging  in  time, 
 \be \label{average}
\llangle f \rrangle =  {\mathbb E}  \, \frac1{\sigma} \int_T^{T+\sigma} f(s) \,ds, 
\ee
 where  $ T\ge1$ and $\sigma\ge \sigma_*$ are parameters of the brackets. 
 \smallskip
 
 Note that  $\llangle f \rrangle $ is the expectation of $f$, regarded as a r.v. 
   on the extended probability space 
 $
 Q= [T, T+\sigma] \times \Omega, 
 $ 
 where the interval $ [T, T+\sigma] $ is provided with the normalised measure $dt/\sigma$. Obviously if $f$ is a stationary process, then 
 $
 \llangle f \rrangle = \EE f(t),
 $
 for any $t$.

  In this notation we have just  proved that 
 $
\llangle \| u^\nu\|_1^2\rrangle \ge \nu^{-1} \,\tfrac12 B_0.
$
But by  Theorem 1 \  $\llangle \| u^\nu\|_1^2\rrangle \le \nu^{-1} C$.    So
\be\label{neww}
\llangle \| u_x^\nu \|_{0}^2\rrangle =
\llangle \| u^\nu \|_1^2\rrangle \sim \nu^{-1}, 
\ee
where $\sim$ means that the ratio of two quantities is bounded from below and from above uniformly in $\nu$ and in the constants 
 $T\ge1$, $\sigma\ge \sigma_*$,
entering the definition of the brackets. A similar  relation  for 3d turbulence 
is one of the  basic  postulates of K41. See  \cite[Section~23]{LL},    \cite[Section~7.3]{F} and \cite[Section~2.2.2]{Fal} 
for more detailed explanation.

Now the 
 Gagliardo-Nirenberg   inequality (Example \ref{ex1}.C)  and   Oleinik's estimate  imply:
 $$
\llangle \| u_x^\nu\|_{0}^2\rrangle
\overset{G-N}{ \le} C'_m \llangle \| u^\nu\|_m^2\rrangle^{\frac{1}{2m-1}}  \llangle | u^\nu_x|_1^2\rrangle^{\frac{2m-2}{2m-1}} \overset{Oleinik}{ 
\le} C_m \llangle \| u^\nu\|_m^2\rrangle^{\frac{1}{2m-1}} .
$$
 From here and \eqref{neww} follows 
  a lower bound for $ \| u^\nu\|_m$: \ 
 $
\llangle \| u^\nu\|_m^2\rrangle \ge C''_m \nu^{-(2m-1)}$
for all  $m\in \N.
$
Combining it with the upper bound in Theorem 3.3, we get:

\begin{theorem}\label{t4.1}
   ({\it second moments of  Sobolev norms of solutions}).  Let $m\in\N$. Then for     any $u_0$ in $H^1$, 
    \be\label{main_est}
      \llangle \| u^\nu(t;u_0) \|_m^2\rrangle
       \sim 
       \nu^{-(2m-1)} .
  \ee
  \end{theorem}

  For $m=0$ behaviour of the norm $\|u\|_0$ is different: 
  \be\label{m=0}
  \lann \|u^\nu\|_0^2\rann \sim 1.
  \ee  
  The upper bound $ \lann \|u^\nu\|_0^2\rann \le C$ follows from Theorem~\ref{t3.2}.ii). Proof of the lower bound see in
  \cite[Theorem~2.3.15]{BK}.
  \smallskip

  As we will see in next sections,   Theorem \ref{t4.1}   and the  Oleinik  estimates  are  powerful and 
 efficient tools  to study  turbulence in  1d Burgers equation \eqref{B}.

 Relation \eqref{main_est} immediately implies that 
 $\ 
 \frac{\ln  \llangle \| u^\nu\|_m^2\rrangle}{ \ln \nu^{-1}} = 2m-1 +o(1)$
 as 
 $
 \nu\to0.
 $
 But can  \eqref{main_est}  be improved to a real asymptotic for $ \llangle \| u^\nu\|_m^2\rrangle$?

\noindent {\it 
  Open problem. }Prove (or disprove)  that for $m\in\N$, 
  $  \llangle \| u^\nu\|_m^2\rrangle$ admits an asymptotic expansion: 
$$
   \llangle \| u^\nu\|_m^2\rrangle = C_m  \nu^{-(2m-1)}  +o( \nu^{-(2m-1)} ) \quad \text{as}\;\; \nu\to0, 
 $$
 for some $C_m>0$. 
 \smallskip
 
 If $u^\nu(t,x)$ is a stationary  in time solution of \eqref{B}, \eqref{xi},
  then from the energy balance \eqref{en_bal} we get that 
 $
  \llangle \| u^\nu\|_1^2\rrangle = B_0 \nu^{-1}.
 $
 So in the stationary case asymptotic  above  is valid  for $m=1$ with $C_1=B_0$.  This is the only situation when we know its validity.

\setcounter{equation}{0}

\section{Turbulence in 3d and 1d}\label{s_burgulence} 

 Now we will  discuss  main heuristic   laws of turbulence and  the 
  K41 theory in their relation  with rigorous results  for  1d turbulence, described by 
   the stochastic Burgers equation \eqref{B},~\eqref{xi}. The results for  1d turbulence 
    will be derived  from the theorems, obtained above.    In this section   we assume that
\be\label{we_assume}
\text{ $u(t) = u(t,x;u_0)$ where $u_0\in H^1$ is a r.v., independent of  } \xi.
\ee
 While speaking about 
K41 we always assume that the corresponding  flows  $u(t,x)$  are 1-periodic in each $x_j$, stationary in $t$,
homogeneous in $x$ 	and satisfy \eqref{Kenergy} and \eqref{Krate}.

 \subsection{Dissipation scale.}\label{ss_9.1}
 The essence of the K41  theory is  analysis of properties of  turbulent  flows  that
are uniform in small values of viscosity $\nu$ (in our choice of units, then the Reynolds number equals $\nu^{-1}$). 
Crucial for    Kolmogorov's  analysis are the  concepts of dissipation and inertial ranges.  The  {\it dissipation range} 
is the region of   wavelengths corresponding to  predominance of the dissipation term in the 
Navier-Stokes equations, while the complementary  {\it inertial range} is characterised by  predominance of the inertial term. 
The two ranges may be defined in the Fourier-  or in  $x$-presentation.  For our  purposes  we choose the first  option.
\medskip

 \noindent {\it 
   In  1d turbulence.}
  Now we present a rigorous theory   of the  ranges in the  context of   Burgers equation (\ref{B}).
  Let us write its solution   $u(t,x)$ as Fourier series 
$$
u(t,x)=
{ \sum}_{s= \pm1, \pm2, \dots} \hat u_s(t) e^{2\pi i sx}. 
$$
 The ranges of $u$ as a 1d turbulent flow are defined via   its {\it dissipation scale}  $l_d$, a.k.a. 
      {\it inner scale}. 
      We define it in the Fourier     representation
 as  the biggest number of the form 
 $$
 l_d(\nu)=\nu^{\gamma}, \quad \gamma>0
 $$
 (corresponding to the smallest possible exponent $\gamma>0$),  such that for $|s| >l_d^{-1}$ the second moments of 
 Fourier coefficient  $\llangle | \hat u_s(t)|^2\rrangle $ 
   as a function of $|s|$ decays    fast,  {\it uniformly in $\nu$.}

   To be precise, for a solution $u=u^\nu(t,x)$ of \eqref{B}, \eqref{xi} 
      let $\Gamma$ denote the set of all real numbers $\gamma'>0$ such that
   $$
   \forall N\in \N \; \; \exists  \ C >0 \;\; \text{such that} \;\;  \llangle | \hat u_s|^2\rrangle\le C  |s|^{-N} \;
    \; \forall\, |s| \ge \nu^{-\gamma'}, \;\;
   \forall \nu \in (0,1]\,,
   $$
   where $C$ depends only on $N$ and $\gamma'$. 

   \begin{defin}
   Mathematical dissipation scale     $l_d=l_d(\nu)$
   (if it exists) equals  $\nu^{\gamma}$, where $\gamma>0$ is defined as 
   $
   \gamma = \inf \Gamma.
   $
      If $\Gamma=\emptyset$ or   $ \inf \Gamma=0$, then  $l_d$ is not defined.  
    \end{defin}

 We emphasize that  $l_d$ depends on $\nu$.   So 
    the concept of {\it dissipation scale} concerns the 
   {\it family of solutions} $u=u^\nu(x,t;u_0)$ for eq.~(\ref{B}),  parameterized by $\nu$ (and depending on  a fixed 
    initial state $u_0$).   Note also that  
   $
   l_d=l_d(\nu)$  vanishes with $\nu$.

      Theorem \ref{t4.1} relatively easy  implies (see in \cite{BK}):
      
      \begin{theorem}\label{t4.2}    The  dissipation  scale $l_d$  of any  solution    $u^\nu(t; u_0)$ of \eqref{B}  equals $\nu$.
     \end{theorem}

 In physics,   dissipation scale   is defined modulo a constant factor, so for eq.  \eqref{B} the physical dissipation scale is 
$C l_d=C\nu$.

  For the Burgers equation, Burgers himself predicted the correct value of   dissipation scale  as $C\nu$.
 \medskip

 \noindent{\it In 3d turbulence}.  In K41 the hydrodynamical (Kolmogorov's) 
 dissipation scale is predicted to be  $l_d^K =C\,\nu^{3/4}$ (we recall \eqref{Krate}). 
  \smallskip

 Dissipation and inertial ranges  are zones, specifying the sizes of  involved
  space-increments. They are defined in terms of the physical 
 dissipation scale.   Namely,   the {\it dissipation range}  (in the $x$-representation)    is the interval
  $
  I_{diss} = [0, c l_d],
  $
  and the {\it inertial range} is    the interval
  $
  I_{inert} =[c l_d, c_1].
  $
    The constants $c$ and $c_1$ certainly do not depend on $\nu$, and   for  1d turbulence, described by eq.~\eqref{B}, \eqref{xi}, 
    depend only on the random force \eqref{xi}. 
  These constants may change from one group of results to another. 
Accordingly in K41 the Kolmogorov inertial range is 
$$
I^K_{inert} = [C l_d^K, C_1] = [C' \nu^{3/4}, C_1],
$$
and the dissipation range is \ 
$
I^K_{diss} =  [0,C' \nu^{3/4}].
$
  
Theorem \ref{t4.3} below implies that in the framework of  1d turbulence the dissipation range $ I_{diss}$ may be defined 
  as the largest closed interval to the right from zero,  such that   for all $l\in I_{diss} $ the 
increments $u(t,x+l) -u(t,x)$ ``statistically behave linearly in $l$".

 \subsection{Moments of  small-scale increments.} \label{ss_9.2}
 We recall \eqref{we_assume}. 
 
  {\it In  1d turbulence.}
Small-scale increments in $x$ of a solution $u(t,x)$ are  quantities 
$
u(t,x+l) -u(t,x) $, where $ x\in S^1$ and $ |l| \ll 1.
$
Their absolute  moments with respect to the averaging in $x$ and  brackets $\lann\cdot, \cdot\rann$, are 
$$
\llangle \int | u(t,x+l) -u(t,x) |^p dx \rrangle  =: S_{p,l} = S_{p,l}(u)
,\quad p>0. 
$$
Function $(p,l)\mapsto S_{p,l}$ is called {\it the structure function} (of a solution  $u$).  Naturally, if a solution $u(t,x)$ is stationary 
in $t$ (see below Section~\ref{s_8.3}), then 
$
S_{p,l}(u) = \EE\big( \int | u(t,x+l) -u(t,x) |^p dx \big). 
$
  If in eq. \eqref{B} the force $\xi(t,x)$ is homogeneous in $x$  (see \eqref{hom}) 
  and $u_0^\om$ also is, then the  random 
   solution $u^\om(t,x;u_0)$ is homogeneous in $x$.\footnote{This is easdy; see in \cite{BK} Theorem 1.5.6.}
 In this case 
      $
    S_{p,l}(u) = \llangle | u(t, x+l) - u(t,x)|^p\rrangle, 
    $
    for any  $x \in S^1$. If in addition      $u^\om(t,x;u_0)$  is stationary in time,  then 
   \be\label{hom_sf}
    S_{p,l} (u)= \EE | u(t, x+l) - u(t,x)|^p,
    \ee
    for  any $t$ and $x$.

Structure function   $  S_{p,l} (u)$, calculated for any solution $u$ of \eqref{B}, \eqref{xi} with $u_0\in H^1$,  obeys  the following law:

\begin{theorem}\label{t4.3} If $u= u(t; u_0)$ is a solution as above and 
$0<\nu \le c_*$ for a  sufficiently small $c_*>0$,  then for each $p>0$ there exists $C'_p\ge1$ such that 
 for  $|l|$ in  inertial range   $ I_{inert} = [c_1\nu, {c}]$ 
with suitable $c, c_1>0$  the  structure function $S_{p,l}= S_{p,l}(u)$ satisfies:
\be\label{assert1}
\begin{split}
 {C'_p}^{-1} |l|^{\min(1, p)}\le S_{p,l} \le C'_p |l|^{\min(1, p)}.
\end{split}
\ee
While for $|l|$ in   dissipation range  $ I_{diss} =  [0, c_1\nu]$, 
 \be\label{assert2}
 C_p^{-1}  |l|^p \nu^{1-\max(1,p)} \le
 S_{p,l} \le C_p |l|^p \nu^{1-\max(1,p)}, \;\;  \forall\, p>0. 
 \ee
 The constants $c_*, C_p, C'_p, c, c_1$ depend on the force $\xi$. 
 \end{theorem}

 In  \cite{AFLV} U.Frisch with collaborators obtained   assertion \eqref{assert1}  
 by a convincing heuristic argument. Below in Subsection~\ref{ss_struct} we present its rigorous proof, based on 
  Theorem~\ref{t4.1} on Sobolev norms of solutions  and  Oleinik's estimate.

 The theory of turbulence also is interested in {\it signed moments} of increments of velocity fields, corresponding to the  {\it
 skew structure function }
 \be\label{S_skew}
 S^{s}_{p,l}= S^{s}_{p,l} (u) := 
\llangle \int ( u(t,x+l) -u(t,x) )^p dx \rrangle  
,\qquad  p \in \N.
\ee
Obviously $S^{s}_{p,l}=S_{p,l}$ for an even integer $p$,  but for an odd  $p$ the moment \eqref{S_skew} is different from $ S_{p,l}$.
The first signed moment 
vanishes, and the third moment  $S^{s}_{3, l}$ is of special interest. 
To apply   Theorem~\ref{t4.3} to  $S^{s}_{3,l} (u) $, where $u$ solves 
\eqref{B}, \eqref{xi}, let us consider the quantity 
$$
I_+ = \llangle  \int_{S^1} \big[ (u(t, x+l) -u(t,x))^+\big]^3dx \rrangle .
$$
Since $u(t, x+l) -u(t,x) \le \ds\int_x^{x+l} u_x^+(t,y)\,dy$, then 
$
I_+ \le \llangle (l |u_x^+(t)|_\infty)^3 \rrangle  \le C_1 l^3,
$
where the second inequality follows from Theorem~\ref{t3.2}.i).  As
$
u= -|u| + 2u^+,
$
then by the last estimate  and \eqref{assert1} we have that for $|l|$ in the inertial range, 
$$
-C^{-1}l +2 C_1 l^3 \ge  \llangle \int ( u(t,x+l) -u(t,x) )^3 dx \rrangle  \ge -C l. 
$$
Thus, there exists $c'>0$ (independent of  $\nu$) such that 
\be\label{B45}
S^{s}_{3, l}(u)
\sim -l \qquad \text{for} \quad l \in [c_1\nu, c'], \; \; 0<\nu\le c_*.
\ee
This is a weak form for  1d turbulence of the 4/5-law from the K41 theory which we discuss below in Section~\ref{s_45}. Literally the
same argument shows that relation \eqref{B45} hold for all moments $S^{s}_{p, l}(u)$ with odd $p\ge3$. 
  \medskip
  
  \noindent{\it In 3d turbulence}.  For a 3d velocity field   $u(x)$ the products
  $
  (u_i(x+l) -u_i(x))\frac {l_j}{|l|}
  $
  make a 9-tensor, and the corresponding  hydrodynamical   structure function organises moments or 
  absolute   moments of this tensor field. 
  Experts in turbulence often work with the {\it longitudinal structure function},  defined as 
  \be\label{abs_mom}
     S_{p,l}^{| | } = S_{p,l}^{| | }(u) 
      = \EE\Big ( \Big|  \frac{(u(x+l)- u(x)) \cdot l}{|l|} \Big|^p\Big), \quad p>0 .
  \ee
    Assuming that a velocity   field $u$ of water turbulence is stationary in time and homogeneous in space, 
    the K41 theory  predicts that for $|l|$      in the inertial range 
     \begin{equation} 
   S_{2,l}^{| | }(u) \sim|l|^{2/3}, 
  \ee
  see \cite{LL, F, Fal}. 
  This is the celebrated    {\it 2/3-law of the K41 theory}. 
  In the same time we have seen in (\ref{assert1}) that for  1d turbulence, for $|l|$ in the inertial range,  $S_{2,l} \sim|l|$.
  
    In the K41 papers the 2/3-law was stated in a stronger form:    it was   claimed there that if fluid's velocity field is a homogeneous 
    and isotropic random field on $\R^3$, then  in the inertial range
  \be\label{univ}
   S_{2,l}^{| | }(u)  =C^K (\eps^K | l|)^{2/3} + o ((\eps^K | l|)^{2/3} ),
  \ee
  where $C^K$ is an absolute constant. But then, due to a criticism from Landau, it became clear that the asymptotic above
  cannot hold (at least, with an absolute constant $C^K$). See in  \cite{LL} a footnote at p.\,126, see
   \cite[Section~6.4]{F}  and see below Section~\ref{ss_10.2}.

 \subsection{Proof of  Theorem \ref{t4.3} }\label{ss_struct}
 To prove the theorem  we have to establish two equivalence  relations in 
  (\ref{assert1}) and (\ref{assert2}). That is, to get  two upper bounds, two lower bounds and show that they make two
   equivalent pairs. Noting that $S_{p,l}$ is an even function of $l$, vanishing with $l$,
 we see that it suffices to consider $S_{p,l}$  with $0<l<1$. 
 \smallskip
  
 \noindent {\it Upper bounds.} We denote 
 $
 u(t,x+l)-u(x) =w_l(t,x)
 $
 and assume  first that $p\ge1$. Using H\"older's inequality we get that 
 \be\label{0.9}
 S_{p,l}(u) \le
 \lann \int |w_l(x)|dx \cdot |w_l|_\infty^{p-1}\rann \le  \lann\big( \int |w_l(x)|dx \big)^p\rann^{1/p} \cdot
  \lann  |w_l|_\infty^{p}\rann^{(p-1)/p} =: I\cdot J.
 \ee
 On one hand, since 
 $$
 \int |w_l(x)|dx = 2\int w_l(x)^+dx \le 2 l \sup_x u_x^+,
 $$
 then by Theorem \ref{t3.2}.i) \ $I\le c_p l$. On the other hand, obviously
 $
 J \le l^{p-1} \lann | u_x |^p_\infty\rann^{(p-1)/p}.
 $
 By the Gagliardo--Nirenberg estimate with a sufficiently large $m$ and by H\"older's inequality, 
 \[
 \lann |u_x|_\infty^p \rann^{(p-1)/p} \le 
 \big[ c \lann \|u\|_m^{2p/(2m-1)} |u|_{1,1}^a\rann \big] ^{(p-1)/p} \le 
 c  \lann \|u\|_m^2\rann^{p/(2m-1)}  \lann  |u|_{1,1}^b\rann^c,
 \]
 for some constants $a,b,c>0$. Now by Theorems  \ref{t3.2}.ii) and \ref{t4.1}
 $
 J \le C l^{p-1} \nu^{-(p-1)}. 
 $
 The obtained bounds on $I$ and $J$ imply the required upper bounds if $p\ge1$. For $p\in (0,1)$ we use H\"older's inequality 
 and the just established bound with $p=1$ to get that 
 $
 S_{p,l} \le \lann \int |w_l(x)|dx\rann^p  =S_{1,l}^p \le (c' l)^p. 
 $
 Since we imposed no restriction  on $l$, then we have thus proved the upper bound in \eqref{assert1}, and that in \eqref{assert2} 
 if $p\le1$. 
 
 To get the upper bound in  \eqref{assert2} for $p>1$, we once again use \eqref{0.9} and again estimate $I$ by $c_pl$.
 While to estimate $J$ we note that since 
 $
|w_l(x)| \le | u_x|_1, 
 $
 then by  Theorems~\ref{t3.2}.ii), 
 $
J \le \lann | u_x|_1^p\rann^{(p-1)/p} \le c_p.  
 $
 This yields that $S_{p,l} \le \tilde c_p l$ and completes the proof of the upper bounds, claimed in Theorem~\ref{t4.3}.
  \smallskip
 
 \noindent {\it Lower bounds.}  We restrict ourselves to the most important case of the lower bound in \eqref{assert1}
 for $S_{p,l}$ with $l$  in the inertial 
 range, $|l|\in [c_1\nu, c]$, when $p\ge1$. The lower bound in \eqref{assert1} for $p<1$ follows from that with $p=1$ and H\"older's 
 inequality, while the proof of the lower bound for \eqref{assert2} is similar to that for \eqref{assert1} but  easier. See  in \cite{BK}.

 Recall that   the brackets 
  $\llangle\cdot\rrangle$ is the  averaging on the extended probability  space  
  $Q=[T, T+\sigma]\times \Omega $, given the measure 
  $
  \rho = (dt/\sigma) \times P. 
  $
  Theorems~\ref{t3.2} and \ref{t3.3} imply the assertion of the lemma below (see \cite{BK} for a non-complicated proof): 
 
 \begin{lemma}\la{l9.5}  There is a constant  $\alpha >0$  and, for any $\nu$
 there is an event $Q_{\nu}\subset Q$     such that\\
  i)  $\rho(Q_{\nu}) \ge \alpha $, and  \\
ii)  for every $(t, \omega)\in Q_{\nu}$ we have 
  \be\label{Q1}
  \alpha  \nu^{-1/2} \le \| u^\om(t)\|_1 \le    \alpha ^{-1} \nu^{-1/2}.
  \ee
   \end{lemma}
   
   Now  for an $M\ge1$ we    define $\bar Q_\nu \subset Q_\nu$ as an event, formed by all $(t,\omega) \in Q_\nu$ such that 
   \be\label{Q2}
 | u_x^{\omega\, +}(t)|_\infty +  | u_x^{\omega}(t)|_1 + \nu^{3/2} \| u^\omega(t)\|_2  + \nu^{5/2} \| u^\omega(t)\|_3 \le M.
   \ee
   By  Theorems~\ref{t3.2}, \ref{t3.3}  and Chebyshev's  inequality, if $M$ is sufficiently large in terms of $\alpha$, then 
   $$
   \rho( \bar Q_\nu) \ge \tfrac12 \alpha. 
   $$
   We fix this choice of $M$ and of the corresponding  events $\bar Q_\nu$. 
   For any fixed $(t,\omega) \in \bar Q_\nu$ let us denote
   $
   v(x) := u^\omega(t,x).
   $
   Below we prove that 
   \be\label{Q3}
   s_{p,l}(v) := \int | v(x+l) -v(x)|^p dx \ge C l^{\min(1, p)}, \qquad \forall \, l\in [c_1'\nu, c_2'],\;\; \nu\in (0, c_*],
   \ee
   with  $C= C(c_1', c_2', p)>0$, if $c_1'$ is sufficiently large and 
   if $c_*$ and $ c_2'>0$ are sufficiently small. Obviously \eqref{Q3}
    (valid for all $(t,\omega) \in \bar Q_\nu$)  implies the required lower    bound. Below deriving  estimates  we systematically 
    assume that 
    \be\label{Q4}
    c_1' \gg1, \; c_2' \ll1,\;  c_* \ll1,
     \quad l \in [c_1'\nu, c_2'].
    \ee
    
    Due to \eqref{Q1} and \eqref{Q2},
    $$
    \alpha^2 \nu^{-1} \le \int |v_x|^2dx \le | v|_{1, \infty} |v|_{1,1} \le M  | v|_{1, \infty} .
    $$
    So
    \be\label{Q5}
     | v|_{1, \infty} \ge \alpha^2 M^{-1} \nu^{-1} =: \alpha_2 \nu^{-1}.
    \ee
    Since $|v_x^+|_\infty \le M$, then $|v_x^+|_\infty  \le \frac12 \alpha_2 \nu^{-1}$ (we recall \eqref{Q4}).  From here
    and \eqref{Q5}, 
    $$
    |v_x^-|_\infty  = |v_x|_\infty  \ge \alpha_2\nu^{-1}. 
    $$
    Let $z \in S^1$ be any point, where $ v_x^-(z)  \ge \alpha_2\nu^{-1}. $ Then 
    \be\label{Q6}
    s_{p,l}(v) \ge
    \int_{z-l/2}^z \Big| \int_x^{x+l} v_y^-(y)dy - \int_x^{x+l} v_y^+(y)dy \Big|^p dx.
    \ee
  By the Gagliardo--Nirenberg inequality and \eqref{Q2}, 
  $$
  |v_{xx}|_\infty \le C_2 \| v\|_2^{1/2}  \|v\|_3^{1/2} \le C_2 M \nu^{-2}. 
  $$
  Since $v^-_x(z) \ge \alpha_2 \nu^{-1}$, then for any $y \in [z- \alpha_3 \nu, z +\alpha_3 \nu]$, where 
  $\alpha_3 = \alpha_2/4C_2 M$, we have 
  $$
  v_x^-(y) \ge \alpha_2 \nu^{-1} -\alpha_3 C_2 M \nu^{-1} = \tfrac34 \alpha_2 \nu^{-1}. 
  $$
   
   Assume that $c_1'\ge\alpha_3$  (cf. \eqref{Q4}). Then $l\ge \alpha_3 \nu$. Since 
    $v_x \le M$ for all $x$, 
    then for $x \in [z-l/2, z]$ we have  that 
   $$
    \int_x^{x+l} v_y^-(y)dy  \ge \int_z^{z+\alpha_3\nu/2} v_y^-(y)dy \ge \frac38 \alpha_2 \alpha_3 \quad \text{and} \quad 
     \int_x^{x+l} v_y^+(y)dy  \le Ml. 
   $$
   So by \eqref{Q6},
   $$
    s_{p,l}(v) \ge   \int_{z-l/2}^z  \big| \tfrac38 \alpha_2 \alpha_3 -Ml |^pdx \ge \tfrac{l}2 \big( \tfrac14 \alpha_2 \alpha_3 )^p,
   $$ 
   if $l \in [\alpha_3\nu,   \alpha_2 \alpha_3 /8M]$ and $c_*$ is sufficiently small. 
   This proves \eqref{Q3} and thus establishes the desired 
   lower bound.

 \subsection{Distribution of energy along the  spectrum.}\label{ss_8.4}
 {\it In  1d turbulence.} 
 For a solution $u(t,x)$ of the stochastic Burgers equation,  regarded as the velocity of a 1d flow, consider the halves of the second moments of its Fourier coefficients 
 $
 \frac12 \llangle |\hat u_s|^2 \rrangle.
 $
 By Parseval's identity, 
      $$
      \llangle  \int  \tfrac12 |u|^2dx\rrangle = \sum_s  \tfrac12   \llangle |\hat u_s|^2\rrangle.
       $$
       So  quantities $ \tfrac12  \llangle |\hat u_s|^2\rrangle$ describe  distribution of the averaged 
       energy  of the flow   along the spectrum.  Another 
       celebrated law of the  K41 theory deals with similar quantities, calculated for 3d turbulent flows; we will return to this below.

  For any ${\mathbf k} \ge1$ define $E_{\mathbf k}(u)$   as the averaging of  $\tfrac12 \llangle |\hat u_s|^2\rrangle$  in $s$ from 
   the layer $J_k^M$ 
  around $\pm {\mathbf k}$, where 
  $\ 
  J_{\mathbf k}^M =\{ n\in \Z^* : M^{-1} {\mathbf k} \le |n| \le M {\mathbf k}\}, $ \ $M>1.
  $
  I.e.  
  \be\label{ek}
  E_{\mathbf k}^B(u)= \llangle e_{\mathbf k}^B(u)\rrangle, \quad e_{\mathbf k}^B(u) = 
  \frac1{| J_{\mathbf k}^M|} {\sum}_{n\in J_{\mathbf k}^M} \tfrac12 |\hat u_n|^2.
  \ee
  The function ${\mathbf k} \mapsto E_{\mathbf k}^B$ is  the {\it energy spectrum} of $u$.

   It is immediate from the definition of $l_d$    that  for ${\mathbf k} $ with ${\mathbf k}^{-1}$ in the dissipation range,  $E_{\mathbf k} ^B(u)$  decays 
  faster than any negative degree of ${\mathbf k} $ (uniformly in $\nu$): for any $N\in\N$, 
      $$
     E_{\mathbf k} ^B \le C_N {\mathbf k} ^{-N} \quad \text{if } \; {\mathbf k}  \gg l_d = C\nu^{-1}. 
       $$
   But for ${\mathbf k}^{-1} $ in the inertial range the behaviour of $E_{\mathbf k} ^B$ is quite different:

 \begin{theorem} \label{t4.4}
     There exists $M'>1$ such that if in the definition of energy spectrum we use layers $J_{\mathbf k} ^M$ with 
 $M\ge M'$, then for    ${\mathbf k}\ge1 $  with ${\mathbf k} ^{-1}$ 
  in the inertial range $I_{inert}$,   i.e. for  $1\le {\mathbf k} \le  c^{-1}  \nu^{-1}$, we have:
  \begin{equation} \label{spec-law}
  E_{\mathbf k} ^B (u^\nu) \sim  {\mathbf k} ^{-2}. 
\ee 
\end{theorem}

  For solutions of the Burgers equation,  Burgers  already in   1948  \ci{B1948}
  predicted that
    $E_{\mathbf k} ^B \sim  {\mathbf k} ^{-2}$ for $1\le {\mathbf k} <\, $Const$\, \nu^{-1}$, 
 i.e. exactly the  spectral power law above.
     \smallskip

 \noindent {\it
 Open problem.} Is the assertion of the theorem true for any $M>1$ if ${\mathbf k} \ge {\mathbf k} _0$ with a suitable ${\mathbf k} _0(M)\ge1$, independent of 
 $\nu$?
 \smallskip

   \noindent {\it Proof of the theorem}.  
     We have to show that 
      \be\label{bounds}
C {\mathbf k} ^{-2}\ge     E_{\mathbf k} ^B \ge C^{-1} {\mathbf k} ^{-2}
       \ee
       for some $C>1$.

       1) {\it Upper bound.}  For any function $u(x)$ we have 
      $$
\hat u_k =  \int u(x) e^{-2\pi ikx} dx = 
\frac1{2\pi ik }  \int u'(x) e^{-2\pi ikx} dx.
       $$
     So from Theorem \ref{t3.2}.ii) we get that 
       \begin{equation} \label{up_bound}
\llangle | \hat u_k|^2 \rrangle \le   \Big( \frac1{2\pi k } \Big)^2 \llangle |u_x|_1^2\rrangle \le C|k|^{-2}. 
      \ee
      This implies the upper bound in \eqref{bounds}.

 2) {\it Lower bound.}  Consider 
 $
 \Psi_{\mathbf k}  = \sum_{|n| \le M{\mathbf k} } |n|^2 \lann | \hat u_n|^2 \rann.
 $
  Since $|\alpha| \ge |\sin\alpha|$, then 
 \[ \begin{split}
 \Psi_{\mathbf k} & \ge \frac{{\mathbf k} ^2}{\pi^2}  \Big( \sum_n \sin^2(n\pi {\mathbf k} ^{-1}) \lann\un\rann -
   \sum_{ |n| >M{\mathbf k} } \sin^2(n\pi {\mathbf k} ^{-1}) \lann\un\rann\Big)\\
 &\ge \frac{{\mathbf k} ^2}{\pi^2}  \Big( \sum_n \sin^2(n\pi {\mathbf k} ^{-1}) \lann\un\rann -  \sum_{ |n| >M{\mathbf k} } \lann\un\rann\Big).
 \end{split}
 \]
 From other hand, by Parseval's identity, 
      $$
\| u(t, \cdot +{{\mathbf k} ^{-1}}) - u(t, \cdot) \|^2 =4 \sum_n \sin^2(n\pi {{\mathbf k} ^{-1}}) | \hat u_n(t)|^2.
       $$
     Taking the brackets $\lann\cdot\rann$ from this equality we see that 
     $$
     S_{2, {\mathbf k} ^{-1}} (u)= 4 \sum_n \sin^2(n\pi {{\mathbf k} ^{-1}}) \lann | \hat u_n(t)|^2\rann
     $$
     (where $S$ is the structure function). Thus, 
     \be\label{d1}
     \Psi_{\mathbf k}  \ge   \frac{{\mathbf k} ^2}{\pi^2}  \Big( \frac14  S_{2, {\mathbf k} ^{-1}}  -  \sum_{ |n| >M{\mathbf k} } \lann\un\rann\Big).
     \ee
     From \eqref{up_bound} we have 
     $
       \sum_{ |n| >M{\mathbf k} } \lann\un\rann \le C_1 M^{-1} {\mathbf k} ^{-1}. 
     $
       Using in \eqref{d1} this estimate jointly with \eqref{assert2} where $p=2$ and $l=1/{\mathbf k} $ we find that 
       $$
       \Psi_{\mathbf k}  \ge  {\mathbf k} ^2 C_2 {\mathbf k} ^{-1} - C_3  M^{-1} {\mathbf k}  =
        {\mathbf k} (C_2 -  M^{-1} C_3) \quad \text{if}\quad c^{-1} \le {\mathbf k}  \le C_1^{-1} \nu^{-1}. 
       $$
       Noting that 
       $$
       E_{\mathbf k} ^B \ge \frac1{2{\mathbf k} ^3 M^3} \sum_{ M^{-1} {\mathbf k}  \le |n| \le M{\mathbf k} }  |n|^2 \lann \un\rann 
        \ge \frac1{2{\mathbf k} ^3 M^3} \Big( \Psi_{\mathbf k}  -   \sum_{ |n| < M^{-1}{\mathbf k} }  |n|^2 \lann \un\rann  \Big)
       $$
       and that by  \eqref{up_bound} the sum  in $ { |n| < M^{-1}{\mathbf k} } $\  in the r.h.s. 
        above is bounded by $C_4 M^{-1} {\mathbf k} $, we arrive at the relation 
        $$
         E_{\mathbf k} ^B \ge \frac1{2{\mathbf k} ^3 M^3}\big( {\mathbf k} (C_2 -M^{-1} C_3) - C_4M^{-1} {\mathbf k} \big).
        $$
        The latter implies the lower bound in \eqref{bounds} if $M$ is sufficiently large and ${\mathbf k} \ge C^{-1}$.  While for 
        $1\le {\mathbf k}  <  C^{-1}$  the bound follows from \eqref{m=0} and \eqref{up_bound} if $C$ and $M$ are sufficiently big.
                \qed
         \medskip
         
         \noindent{\it In 3d turbulence.}  Let us consider a 1-periodic in space turbulent flow $u(t,x)$ of water
         with Fourier coefficients $\hat u(t,s)$. Next for $r\ge1$ denote by $E_r^K$ the averaging of energies $\frac12 | \hat u(t,s)|^2$ over $s$ from  a suitable 
         layer around the sphere $\{ |s| =r\}$ and in ensemble.    The celebrated  Kolmogorov--Obukhov law predicts that 
  $$
     E_r^K \sim  r^{-5/3} \quad \text{ for \  $r^{-1}$ \  in  Kolmogorov's inertial range  } \  [C'\nu^{3/4}, C_1]. 
 $$
see \cite{LL, F}.

          \noindent {\it  Open problem.}   We saw that the Oleinik estimate and theorem on moments of small-scale increments of solutions for \eqref{B} jointly imply the  spectral power law of  1d turbulence. 
          Under what  assumption on $u(x)$ the latter is equivalent to the  theorem on moments of small-scale increments?
 More interesting is this question, asked for the laws of K41: under what restriction on a field $u(x)$ the Kolmogorov 2/3-law is
 equivalent to the Kolmogorov--Obukhov law?  Or at least one of them implies another?
 See Section 3.4 of \cite{K21}  for a discussion.

  \setcounter{equation}{0}
\section{Statistical equilibrium (the mixing)} \label{s_8.3}

 It is a general believe in the theory of 3d turbulence that  as time grows, statistical 
characteristics of a turbulent flow $u(t,\cdot)$ converge to a universal statistical equilibrium. 
 E.g. see in \cite{Bat} pages 6-7 and 109. 
 Mathematically it means that if we regard a space-periodic turbulent flow $u(t,x)$ 
as a random process $u(t, \cdot)$ in some function space $\cH$ of 1-periodic non-compressible 3d  vector fields, then for any bounded continuous 
functional $f$ on $\cH$  we have 
\be\label{mix} 
{\mathbb E} f(u(t, \cdot)) \to \lan f,  \mu_\nu\ran \quad \text{as} \quad t\to\infty,
\ee
where $\mu_\nu$ is a measure on $\cH$, describing the  equilibrium.
The property, manifested by relation \eqref{mix} is called  { \it the mixing}, see \cite{DZ, KS}.  Since
 the K41 theory deals with stationary in time turbulent vector 
fields, then there  convergence \eqref{mix} trivialises to an equality which holds for all $t$. 

 In  1d turbulence, if 
$u(t,x)$ is a solution of \eqref{B} and $\cH$ is the space $H^m$ with some $m\in\N$, then the  validity of  convergence \eqref{mix} 
with a suitable measure  $\mu_\nu$ on $H^m$ may be  derived from  general results for SPDEs (e.g. see in \cite{KS}). 
But then the rate of  convergence would 
depend on $\nu$.   At
 the same time,  in the theory of turbulence it should not depend on $\nu$,  and, as we show in this section,  for eq.~\eqref{B} it does not!

 \subsection{Convergence in distribution for solutions with different 
  initial states} 

 The result below is a key step in establishing  the mixing  (\ref{mix}) for solutions of the stochastic Burgers equation.

   \begin{theorem}\label{t4.5}
  Let  $f$ be a  continuous functional on  the space $L_1(S^1)$ such that 
  \be\label{f}
  \text{Lip}\, f \le1, \quad |f|\le1.
  \ee
  Then for any $u_1, u_2\in H^1$ and $t\ge3$ we have 
\be\label{mixt}
\big| \EE  f(u(t;u_1)) - \EE f(u(t;u_2))  \big| \le C (\ln t)^{-1/8},
\ee
  where $C>0$ depends    on the force $\xi$, but does not depend   on $f$, $\nu$,     $u_1$ and $u_2$.
 \end{theorem}
 
 \begin{proof}  The proof follows from a combination of three results: 
 i) a lower 
 bound for the probability that during a time $T$ the Wiener process $\xi(t)$ stays inside the ball $B^\eps_{H^m}$ with any  small 
  $\eps>0$.
ii)  
 the $L_1$-nonexpending property of eq.~\eqref{B}, stated in Lemma~\ref{lnexp}, 
 and  iii) the Oleinik maximum principle.
 \smallskip
 
  \noindent {\it 1. \!Lower bound for the probability.}  The probability in question is 
  $$
  \PP\{ \| \xi \|_{X_T^m}<\varepsilon \} =: \gamma^m_{\varepsilon, T}.
  $$
 Let us consider the function
 $$
 f_m(\varepsilon)= e^{-\kappa_m(\varepsilon^{-3} + \varepsilon^{-2} )},\qquad \varepsilon>0,
 $$
 where $\vka_m>0$ is chosen in the next lemma.
 \bl\la{lgam} For any $m\ge 1$, there exists a $\kappa_m>0$ (depending on the force $\xi$) 
  such that
 \be\label{h0}
\gamma^m_{\varepsilon, T} \ge  \tfrac12 f_m(\varepsilon/\sqrt{T}) \qquad \forall\, 0<\varepsilon\le1,\; \forall\, T>0.
\ee
 \el
 \bpr
 Since $\cD \xi(t) = \cD \big(\sqrt{T} \xi(t /{T})\big) $, then 
 $
 \gamma^m_{\varepsilon, T} = \gamma^m_{\varepsilon/\sqrt{T}, 1}.
 $
 So  it suffices to prove the estimate with   $T=1$. Let us denote  
$$
 \xin(t)=  \Pi_N \xi(t), \qquad \xinn(t) = \xi(t)-\xin(t).
$$
Then 
$$
\gamma_{\varepsilon,1}^m \ge   \PP\{ \| \xin \|_{X_T^m}<\varepsilon/\sqrt2 \} \cdot  \PP\{ \| \xinn \|_{X_T^m}<\varepsilon/\sqrt2 \} =: P^N \cdot\,  ^N\!P.
$$
Estimating\,  $ ^N\!P$ is easy. Indeed, by Chebyshev's inequality 
$
 ^N\!P \ge 1-2\varepsilon^{-2}  \EE \| \xinn \|^2_{X^m_1}. 
$
Since by Doob's inequality
$
  \EE \| \xinn \|^2_{X^m_1} \le 4   \EE \| \xinn(1) \|^2_{m}, 
$
then 
$$
^N\!P\ge 
1-8 \varepsilon^{-2} \EE \| \xinn(1)\|_m^2 \ge 1- 8 \varepsilon^{-2} (2\pi N)^{-2} B_{m+1},
$$
where  the  second  inequality follows from Proposition \ref{p1}.ii) and the definition of $\xinn$. Choosing
$
N=N_\varepsilon = [2\sqrt{B_{m+1}}/\pi\varepsilon] +1
$
we achieve that 
$$
^N\!P  \ge 1/2.
$$ 

Now we  estimate $P^N$. First we note that  for any vector $\xi = \sum_{|s|\le N}  b_s \xi_s e_s$, where $\{ b_s\}$ are the constants in \eqref{xi}, 
 relation 
$
|\xi_s| <  \varepsilon B_m^{-1/2} /\sqrt2=: \varepsilon',
$
if valid for all $ s$,
implies that 
$
\| \xi\|_m^2 <\varepsilon^2/2.
$
So 
$$
P^N \ge \prod_{|s|\le N} \PP \Big\{ \sup_{0 \le t\le1} |\beta_s(t) |  < \varepsilon' \Big\} =\big(\rho(\varepsilon')\big)^{2N},
$$
where $\rho(\varepsilon') =  \PP \big\{ \sup_{0 \le t\le1} |\beta_s(t) |  < \varepsilon'\big\}$. The function $\rho$ is well known in probability. It
is given by a converging series and admits a lower bound 
$
\rho(a) \ge e^{-\pi^2 /8a^2}.
$
See \cite[Section 3.2]{BK} for  discussion and references, and see there Problem~3.2.3 for a sketch of a proof of this estimate.
Thus
$$
P^N \ge \big(\rho( \varepsilon B_m^{-1/2} / \sqrt2)\big)^{2N} \ge e^{-\kappa_m(\varepsilon^{-3} + \varepsilon^{-2} )} = f_m(\varepsilon), 
$$
for an explicit positive constant  $ \kappa_m$. So
(\ref{h0}) is established. 
\epr

  \noindent {\it 2. \!End of the proof.}  a) For an $N\in\N$  let us cut $[0,\infty)$ to blocks of $G(N)$ segments of length $N$, where 
  $G(N) = C\exp N^8$ (so each block  itself is a segment of length $NG(N)$). Then for a suitable $C$ the probability of event
  $$
  \Gamma_N =\{ \omega: \text{for each} \; 1\le k\le G(N),  \sup_{(k-1)N \le t\le kN} \| \xi^\om(t) - \xi^\om((k-1)N) \|_4 >N^{-2}/4 \}
  $$
  satisfies 
  \be\label{bound}
\PP(\Gamma_N)\le N^{-1},
\ee
 for each $N$. Indeed, since the increments of $\xi(t)$ on disjoint segments are i.i.d., then 
  by \eqref{h0} 
$
  \PP(\Gamma_N)\le \big(1-f_4(1/4N^{5/2})\big)^G,
  $
  where $f_4$ is the function from Lemma~\ref{lgam} with $m=4$. Now \eqref{bound} follows by an easy calculation. 
  \smallskip
  
  b) Now for $t\ge0$ consider the function 
  $
  F^\om(t) = \sup_{u_0\in H^1}  |u^\om(t; u_0)|_1.
  $
We claim that for each $N\in\N$ and for $G(N)$ as above,
 \be\label{h5}
 \PP\big(
 \inf
  F^\om(kN) >24 N^{-1} \} \le N^{-1},
  \ee
  where the infimum is taken over $k$ from the integer segment $[0, NG(N)] \cap \Z$. 
Indeed, the estimate follows from \eqref{bound} and  Theorem~\ref{t3.1}.iii) with $\theta=T=N$, applied to segments 
$
[(k-1)N, kN], $ $ k=1, \dots, G(N).
$

 \smallskip
  
  c) Now we complete the proof. For any $u_1,u_2 \in H^1$ consider the random process 
  $
  U(t) = (u(t;u_1), u(t;u_2)) \in H^1\times H^1.
  $
  For $N\in\N$ define  closed sets  $O_N \subset H^1$, 
  $
  O_N = \bar B_{L_1}^{24N^{-1}} \cap H^1 ,
  $
  and  hitting times 
  $$
  \tau^\om_N =\min\{ l\in\N:l\le NG(N), \ U^\om(l)  \in O_N\times O_N\},
  $$
  where $ \tau_N =\infty$ if the set is empty. Applying \eqref{h5} 
  to solutions $u(t; u_j)$ we find that 
  $
  \PP(\tau_N=\infty) \le 2N^{-1}.
  $
  But if $\tau_N <\infty$, then 
  $\,
  | u^\om(\tau_N; u_1) -u^\om(\tau_N; u_2) |_1 \le 48N^{-1} ,
  $
  and this inequality still holds for $t\ge\tau_N$ by Lemma~\ref{lnexp}.  So for any functional $f$ as in \eqref{f}, for $t\ge NG(N)$ we have
  $$
  \big| \EE\big( f(u(t;u_1) -  f(u(t;u_2) \big) \big| \le 2N^{-1} +48 N^{-1}. 
  $$
Here the first term in the r.h.s. comes from the integrating over the event $\{\tau_N =\infty\}$ (since $|f|\le1$)
and the second -- over its complement (since Lip$\,f\le1$). 
 Thus the l.h.s. of \eqref{mixt} is at most $50N^{-1}$ if $t\ge NG(N)$ for some $N\in\N$. This implies the assertion of Theorem \ref{t4.5} 
 as  $\log(NG(N)=\log N+N^8\sim N^8$ for large $N$. 
 \end{proof} 
 
 \subsection{The mixing}
 
 We recall that for a  complete  separable metric space  $M$ and  two measures $ \mu_1,\mu_2\in\cP(M)$, 
the dual-Lipschitz distance between $ \mu_1$ and $\mu_2$ (also known as the Kantorovich--Rubinstein  distance) is 
\be\label{dL}
\|\mu_1-\mu_2\|_{L}^* =
\|\mu_1-\mu_2\|_{L,M}^*:=\sup_{f\in C_b^0(M), |f|_L\leqslant1}\Big|\langle f ,\mu_1\rangle-\langle f ,\mu_2\rangle\Big|\le2,
 \ee
   where  $|f|_L= |f|_{L,M}=     \,$Lip$\,f+\|f\|_{C(M)}$. This distance converts $\cP(M)$ to a complete metric space, and
   the convergence in it is equivalent to the weak convergence of measures   (see in \cite{Vil, KS, BK}). 
    Then Theorem~\ref{t4.5} means that 
  \be\label{mixx}
  \| \cD u(t;u_1) - \cD u(t;u_2)\|^*_{L, L_1} \le  C (\ln t)^{-1/8} \quad \text{ for all $u_1, u_2\in H^1$ and  $t\ge3$}. 
  \ee
  
  Now let $\mu=\mu_\nu$ be a stationary measure for the stochastic Burgers equation \eqref{B2},
   considered on the space $H^1$    (see \eqref{KB}), and let 
  $\lambda\in \cP(H^1)$. Let $f$ be a  continuous   function on $L_1$ as in \eqref{f}. Consider 
  $$
  X_t^f =: \big|  \lan f, S^*_t\lambda\ran - \lan f,\mu\ran\big| = \big|  \lan f, S^*_t\lambda\ran - \lan f, S_t^*\mu\ran\big|.
  $$
  We have
  $$
    \lan f, S^*_t\lambda\ran  =   \lan S_t f, \lambda\ran 
    =\int_{H^1} S_tf(u_1) \lambda(du_1) = \int_{H^1}   \int_{H^1} S_tf(u_1) \lambda(du_1) \mu(du_2) 
  $$
  (see \eqref{dual}). 
  Similarly, 
  $
    \lan f, S^*_t\mu \ran  =  \int_{H^1}   \int_{H^1} S_tf(u_2)\mu(du_2) \lambda(du_1).
  $
  Therefore
  $$
  X^f_t  \le  \int_{H^1}   \int_{H^1}  \big|S_tf(u_1) -  S_tf(u_2)  \big| \lambda(du_1) \mu(du_2). 
  $$
  By Theorem \ref{t4.5} the integrand is bounded by $ C (\ln t)^{-1/8}$. So 
  $
  X^f_t  \le C (\ln t)^{-1/8}.
  $
  Since $f$ is any continuous   function, satisfying \eqref{f}, we have proved

    \begin{theorem}\label{t4.55}
    Let $\mu_\nu \in \cP(H^1)$ be a stationary measure as in \eqref{KB} with $m=1$ 
    and $\lambda$ -- any measure from $\cP(H^1)$. Then 
    \be\label{mixL1}
     \| S^*_t \lambda- \mu_\nu \|^*_{L, L_1} \le  C (\ln t)^{-1/8} , \quad t\ge3,
     \ee
     where $C$ is the constant from Theorem \ref{t4.5}.  In particular, a stationary measure $\mu_\nu$ is unique. 
 \end{theorem}
 
 \begin{remark} \label{r_hommeas}
 1) Since for any $m$ a stationary measure $\mu_\nu\in \cP(H^m)$ as in \eqref{KB} 
 also is a stationary measure in $H^1$, then by  uniqueness
  $\mu_\nu$
 does not depend on $m$, and so  $ \mu_\nu\big( H^m\big) =1$   for all $m\in\N$.

 2) It is easy to see that if 
  the random field $\xi(t,x)$ is homogeneous in $x$ (see \eqref{hom}), then  the measure $\mu_\nu$ also is. 
\end{remark}

  \begin{cor}\label{c_mix} For any $1\le p<\infty$ there exists  positive constant $C_p$ 
   depending only on  the force $\xi$, 
  such that under the assumption of Theorem~\ref{t4.55}
      \be\label{mixLp}
     \| S^*_t \lambda- \mu_\nu \|^*_{L, L_p} \le  C_p (\ln t)^{-\kappa_p}  \;\; \text{for} \;\;   \;\;  t\ge3, \;\; \text{where} \;\;  \kappa_p = \tfrac1{9(2p-1)}. 
     \ee
 \end{cor}

 \begin{proof}
 Since for $p=1$ estimate \eqref{mixL1} implies the assertion,  we may assume that $1<p<\infty$. Consider a solution $u(t)$ such that $\cD u(0)= \lambda$
 and denote $\lambda_t = \cD u(t) = S_t^*\lambda$.  In view of Theorem~\ref{t3.2}.ii), 
 $$
 \lan | u|_1^\gamma, \lambda_t \ran = \EE |u(t)|_1^\gamma \le C_\gamma<\infty \qquad \forall\, t\ge1,\; \gamma\ge0, 
 $$
 and similar
  $$
 \lan | u|_1^\gamma, \mu_\nu \ran \le C_\gamma<\infty \qquad \forall\, t\ge1,\; \gamma\ge0.
 $$
 Apart from the dual-Lipschitz distance on $\cP(L_1)$, consider there the Kantorovich distance 
 \be\label{Kant}
 \| \mu-\nu\|_{Kant} = \sup_{f\in C(L_1),\, \text{Lip}\,f\le1} | \lan f, \mu\ran - \lan f, \nu\ran \le\infty. 
 \ee
 Obviously,
 $
 \| \lambda_t - \mu_\nu \|_L \le  \| \lambda_t - \mu_\nu \|_{Kant}.
 $ 
 But in view of the estimates on moments of $\lambda_t$ and $\mu_\nu$ 
  the Kantorovich distance between them may be estimated via the 
 dual-Lipschitz distance. To do this in the r.h.s. of \eqref{Kant} we replace $f$ by $f^R = \min(|f|, R) \sgn\! f$, $R\ge1$. Then
 $| f^R|_L \le R$, so   the
 modified supremum  in \eqref{Kant}   is at most  $R\| \lambda_t - \mu_\nu\|_L$, 
  while the difference between the modified and non-modified supremums    may be 
 estimated in terms of $R$ and high moments of the two measures. Minimising the obtained estimate in $R$ we get that 
 \be\label{2.82}
   \| \lambda_t - \mu_\nu \|_{Kant} \le C   \| \lambda_t - \mu_\nu \|_{L} \le C' (\ln t)^{-1/9}, \quad t\ge3
 \ee
 (see \cite[Section 4.2]{BK} for details of this calculation). By the Kantorovich--Rubinstein theorem (see in \cite{Vil, BK}) 
 there exist r.v.'s $U_t$ and $\tilde U_t$ in $L_1$ such that 
 $\cD U_t = \lambda_t$, $\cD  \tilde U_t = \mu_\nu$ and 
 $
 \EE | U_t - \tilde U_t|_1 = \| \lambda_t - \mu_\nu \|_{Kant} .
 $
Due to this relation and \eqref{2.82},
 $$
 \EE | U_t - \tilde U_t|_1 \le  C' (\ln t)^{-1/9}, \qquad t\ge3.
 $$
 Since $\lambda_t$ and $\mu_\nu$ are measures on $H^1$, then 
 $U_t, \tilde U_t \in L_p$, a.s. For   any $f\in C_b(L_p)$  such that 
 $
 |f|_{L, L_p} \le 1
 $
 (see \eqref{dL})  we have 
  \be\label{zz15}
  \begin{split}
  | \lan f, \lambda_t\ran -  \lan f, \mu_\nu\ran| &=\EE | f\circ U_t -f\circ { \tilde{U_t}}|  \le  \EE | U_t - { \tilde{U_t}}|  _p \le 
   \EE | U_t - { \tilde{U_t}}|  _1^{1-\theta} |U_t - { \tilde{U_t}}|  _{2p}^{\theta}\\
  &\le(\EE | U_t - { \tilde{U_t}}|  _1)^{1-\theta}   \big( \EE(   |U_t| _{2p} +| { \tilde{U_t}}|  _{2p})\big)^{\theta},\quad \theta= {(2p-2)}/{(2p-1)}, 
  \end{split}
  \ee
  where the second  estimate    is the Riesz--Thorin inequality. 
As
$
 \EE(   |U_t| _{2p} +  |{ \tilde{U_t}}|  _{2p} )= \lan |u|_{2p}, \lambda_t\ran +  \lan |u|_{2p}, {\mu_\nu}\ran  \le C_{p},
 $
 then the r.h.s. of \eqref{zz15} is bounded by  $C_p (\ln t)^{-\kappa_p}$ with
 $\kappa_p=1/9(2p-1)$, for all $f$ as above. So \eqref{mixLp} is proved. 
 \end{proof}  
 
 Using instead of Theorem \ref{t3.2} estimates \eqref{*3.3} and  arguing as when proving the corollary 
 above  we may also  get that  under 
 the assumption of Theorem~\ref{t4.55} for any $M\in\N$, 
 $$
  \| S^*_t \lambda- \mu_\nu \|^*_{L, H^M} \le  C_M(\nu) (\ln t)^{-\kappa_M(\nu)},  \quad t\ge3.
 $$
 The dependence of the constants on $\nu$ makes this result less interesting, but still it shows that  equation \eqref{B2},  \eqref{xi} 
  is mixing in every Sobolev space $H^M$.

\subsection{Energy spectrum and structure function of the stationary measure}
Stationary solution $u^{stat}(t)$ of \eqref{B}  is a solution such that 
  $\ 
  \cD  u^{stat}(t) = \mu_\nu \;\; \text{for all} \;\; t,
  $
  where $\mu_\nu$ is the stationary measure.    Energy spectrum of  $ \mu_\nu $  is the function 
$\, 
E_{\mathbf k}^B(\mu_\nu) = \int e_{\mathbf k}^B(u) \mu_\nu(du) 
$
($e_{\mathbf k}^B$ is defined in \eqref{ek}). 
Obviously, 
  $$
  E_{\mathbf k}^B(\mu_\nu)= \llangle e_{\mathbf k}^B(u^{stat} (\cdot))\rrangle  =  \EE  e_{\mathbf k}^B(u^{stat}(t))\quad \forall\, t\ge0.
   $$
   Since $\llangle e_{\mathbf k}^B(u^{stat} (t))\rrangle$ satisfies the spectral power law (\ref{spec-law}),
    then $E_{\mathbf k}^B(\mu_\nu)$ also does:
      \be\label{SPL}
  E_{\mathbf k}^B(\mu_\nu) \sim {\mathbf k}^{-2},\;  \quad 1\le {\mathbf k} \le C_1^{-1} \nu^{-1}. 
   \ee
      The map $u\mapsto \hat u_n$ is a continuous linear  functional on  
   $L_1$     of unit norm. Moreover, all moments of the $L_1$-norm of   a solution $u(t;u_0)$, $u_0\in H^1$, 
   are bounded uniformly in $t\ge1$.  Hence,  from \eqref{mixL1} we get that 
   $$
   \EE e_{\mathbf k}^B(u(t;u_0)) \to E_{\mathbf k}^B (\mu_\nu) \quad\text{as} \quad t\to \infty,
   $$
    where the rate of convergence does not depend on $\nu$ and ${\mathbf k}$.   
   So asymptotically as $t\to \infty$ the instant energy spectrum $ \EE e_{\mathbf k}^B(u(t;u_0)) $ also satisfies the 
   spectral power law. 
   \smallskip
   
   Writing the structure function of a solution $u =u(t;u_0)$ as 
   $$
   S_{p,l} (u) = \llangle  s_{p,l} (u) \rrangle,  \qquad s_{p,l} (v)= \ds\int |v(x+l) -v(x)|^pdx, 
   $$
   we define  $ S_{p,l} (\mu_\nu) $ as $\langle  s_{p,l} , \mu_\nu\rangle$. Similar to the above, $ S_{p,l} (\mu_\nu)$  satisfies the relations in 
    Theorem~\ref{t4.3}.  Noting that $ s_{p,l} $ is continuous on the space  $L_{\max(p, 1)}$ we derive from Corollary~\ref{c_mix} that 
    $
    \EE  s_{p,l} (u(t;u_0)) \to S_{p,l} (\mu_\nu),
    $
    as $t\to\infty$, uniformly in $l$ and $\nu$. 
    So asymptotically as $t\to\infty$ the instant structure function $\EE  s_{p,l} (u(t;u_0))$ also satisfies \eqref{assert1} and \eqref{assert2}.
   \smallskip
   
   As we pointed out (see \eqref{hom} and Remark \ref{r_hommeas}), if  the random force is homogeneous, 
    then stationary measure $\mu_\nu$ also is. In this case 
    $$
     S_{p,l} (\mu_\nu) = \langle |u (x+l)- u(x)|^p, \mu_\nu\rangle = \EE | u^{stat}(t, x+l) - u^{stat}(t, x)|^p, 
    $$
    for any $x$ ant $t$. This is in close agreement with the objects, treated by K41, where velocity fields $u$ are 
    assumed to be stationary and homogeneous (see in  Section~\ref{ss_9.2}).

      \setcounter{equation}{0}
      \section{The 4/5-law and Landau objection} \label{s_45}
      
      In this section we follow  paper \cite{PGK}. Now talking about K41, as in the K41 papers, we assume that the involved 
      velocity fields $u(t,x)$ are homogeneous and       isotropic in $x$ (so are not space-periodic).
      
      \subsection{The 4/5-law} \label{ss_10.1}

      \noindent{\it In 3d turbulence.}  Apart from the absolute moments of longitudinal increments \eqref{abs_mom}, K41 studies their signed 
      cubic moments
      $$
     S_{3,l}^{| | s } = S_{3,l}^{| | s }(u) 
      = \EE \Big(  \frac{(u(x+l)- u(x)) \cdot l}{|l|} \Big)^3,
  $$
   where $u$ is a  velocity field of a turbulent flow. 
  Concerning this quantity K41 makes a very precise prediction, called   {\it Kolmogorov's  4/5-law}:
\be\label{K45}
 S_{3,l}^{| | s } = -(4/5) \eps^K |l| + o(\eps^K |l|) 
  \quad \text{when \ $l$ \ is in the inertial range},
\ee
where $\eps^K$ is the rate of  energy dissipation \eqref{Krate}  (which is $\sim1$ by 
assumption). The law was intensively discussed by 
physicists and was re-proved by them a number of times, using physical arguments, related to that in the K41 papers. Recently
a  progress in rigorous verification of the law was achieved in \cite{Bedr}. There the relation in \eqref{K45} is established for
stationary solutions  $u(t,x)$ of the stochastic 3d~NSE on a torus, assuming that they meet the assumption 
$
\nu \EE |  u|_{L_2}^2 =o(1)
$
as $\nu\to0$, and that $| l |$ belongs to  some interval in $\R_+$ whose left edge converges to 0 with $\nu$, but whose
relation with the inertial range is not clear.
\medskip

  \noindent{\it In  1d turbulence.}   Following K41, proofs of the 4/5-law in physical works, as well as in the 
  rigorous paper \cite{Bedr}, crucially use the Karman--Howard--Monin formula (rather a class of formulae with this name). For a flow $u(t,x)$ the 
  formula relates time-derivative of $S_{2,l}^{| |}(u(t, \cdot)) $ with derivatives of $S_{3,l}^{| | s }(u(t, \cdot))  $ in $l$. 
  Variants of this formula,    e.g.  that
  in \cite{Fal},  instead of second moments $S_{2,l}^{| |} $ use   correlations 
  $
  \EE( u(t,x) \cdot u(t, x+l)), 
  $
  closely related to $S_{2,l}^{| |} $. Thus motivated let us 
   examine time-derivatives of  space-correlations for a solution $u(t) =u(t,x;u_0)$ of \eqref{B} 
 with some  random initial data $u_0\in H^{3}$:
  $$
  f^l(u(t)) := \int u(t,x) u(t,x+l) dx.
  $$
  To simplify a bit notation in this subsection and in the next one we assume that in the 1d case $l>0$.

  Applying Ito's formula to $f^l(u(t))$ (with the same stipulation as   we made in Section \ref{s_4} after Theorem~\ref{t5}) 
   and noting that    $d^2 f^l(u)(e,e)= 2 f^l(e)$ we arrive at the equality
  \[
 \begin{split}
 \frac{d}{dt} \EE f^l(u(t)) = \EE\big( -df^l(u)(uu_x) +\nu\, df^l(u)(u_{xx}) +  \sum b_s^2  f^l(e_s)\big) 
  =: \EE( -I_1(t) +I_2(t) +I_3(t)).
 \end{split}
 \]
 Noting that
 $
 df^l(u)(v) = \int \big( u(x) v(x+l) + u(x+l) v(x))\big) dx
 $
 and that, trivially,
 $
 (\partial/\partial l) u(x+l) = u_x(x+l) ,
 $
 we calculate that 
  \[
 \begin{split}
 I_1(t) &= -\frac16 \frac{\partial}{\partial l} s_{3,l} (u(t)), \quad \;\;s_{3,l} (v(x)) = \int \big( v(x+l) -v(x)\big)^3dx;\\
  I_2(t) &= 2\nu  \frac{\partial^2}{\partial l^2}  f^l(u(t)); \qquad \;\;
   I_3(t) =   \sum b_s^2  \cos(2\pi sl) =: \tilde B_0(l)
  \end{split}
 \]
  (see \cite{PGK} for details).  So
  \be\label{KHM}
  \frac{d}{dt} \EE f^l(u(t))  = \frac16 \EE \frac{\partial}{\partial l} s_{3,l}(u(t)) + 2\nu \EE  \frac{\partial^2}{\partial l^2} f^l(u(t)) + \tilde B_0(l).
 \ee
 This relation is a version of the Karman--Howard--Monin formula for the stochastic Burgers equation. 
 
 Now let $u(t) = u^{st}(t)$ be a stationary solution of \eqref{B}, \eqref{xi} (see Section~\ref{s_8.3}). Then
 the l.h.s. of \eqref{KHM}  vanishes. Since $s_{3,0}=0$ and 
 $
 (\partial/\partial l) f^l(u)\!\mid_{l=0} =0, 
 $
 then integrating \eqref{KHM} in $dl$ and multiplying  it by 6 we find that 
 \be\label{I5}
 \EE \big(s_{3,l}( u^{ st}(t) )\big) = -12 \nu  (\partial /\partial l) \EE \big(f^l (u^{ st}(t) )\big) -6 \int_0^l \tilde B_0( {r}) d {r}.
\ee
By Theorem \ref{t3.3}
 \be\label{I6}
 \EE \| u^{  st} (t)\|_1^2 = \lan \| u\|_1^2, \mu_\nu\ran
  \le C\nu^{-1}. 
 \ee
 Consider the first term in the r.h.s. of \eqref{I5}.  Abbreviating 
 $ u^{  st}$ to $u$ we get: 
  \[
 \begin{split}
 \big|  (\partial /\partial l) \EE \big(f^l (u (t) )\big) \big|&=
\big| \EE \int u(t,x) u_x(t,x+l) dx\big|  =  \big| \EE \int (u(t,x) - u(t,x+l)) u_x(t,x+l) dx\big|  \\
  &\le
 \Big[ \EE \int\big( u(t,x) -u(t,x+l)\big)^2dx\Big]^{1/2}  \Big[ \EE \int  u_x(t,x)^2dx\Big]^{1/2} .
 \end{split}
 \]
 Since $u$ is a stationary solution, then the first factor in the r.h.s. equals $S_{2,l}^{1/2}$. So in view of 
 \eqref{I6} and Theorem~\ref{t4.3} the first term in the r.h.s. of \eqref{I5} is 
 $O(\sqrt{l} \sqrt\nu)$.  As $ \tilde B_0(l)$ is a $C^2$-smooth  even function and $ \tilde B_0(0)= B_0(0)$,   then 
 $
 \int_0^l \tilde B_0(r)dr = B_0 l +O(l^3).
 $
 We have seen that 
 \be\label{I7}
  \EE \big(s_{3,l}( u^{ st}(t) )\big) = -6 B_0 l+ O(l^3)  +  O\big(\sqrt{l} \sqrt\nu\big).
 \ee
 If $l$ belongs to the inertial range $[C\nu, C_1]$, then $ O\big(\sqrt{l} \sqrt\nu\big) \le$Const$\,C^{-1/2} l$. 
So assuming that $C$ is sufficiently big we get from \eqref{I7} another proof of the weak law 
\eqref{B45} for stationary solutions $u^{st}(t)$. 

Now let $l$ belongs to a ``strongly inertial range",
 \be\label{I8}
  l \in [ L(\nu) \nu, C_1], 
   \ee
   where $L$ is some fixed positive function of $\nu>0$ such that 
   $$
   L(\nu)\nu \to0\;\; \text{and}\;\;
  L(\nu) \to \infty \;\;\text{as}\;\; \nu\to0.
   $$
   Then $\sqrt{l} \sqrt\nu = o(l)$ as $\nu\to0$ and we get from \eqref{I7} 
   
   \begin{theorem}\label{t_45str}
   Let $u^{st}(t)$ be a stationary solution of \eqref{B}, \eqref{xi} and $l$ satisfies \eqref{I8}. Then 
   \be\label{I9}
  \EE \big(s_{3,l}( u^{ st}(t) )\big) = -6 B_0 l +  o(l) \quad \text{as} \quad \nu\to0, 
 \ee
 where $o(l)$ depends on the function $L(\nu)$ and the random force $\xi$. 
   \end{theorem}
   
   Due to the balance of energy relation \eqref{patV}, for  stationary solutions $u^{st}(t)$ the rate of energy dissipation 
   $\eps^B = \nu \EE \| u^{st}(t)\|_1^2 $ 
   is given by 
   \be\label{beB}
   \eps^B  = \tfrac12 B_0.
   \ee
   So relation \eqref{I9} may be written as 
    \be\label{bubu}
  \EE \big(s_{3,l}( u^{ st}(t) )\big) =  -12 \eps^B l + o(l) \quad \text{as} \quad \nu\to0.
 \ee
In this form the 4/5-law for  1d turbulence appears in works of physicists, justified by a heuristic argument. 

Combining the last theorem with Theorem \ref{c_mix} we get that 
 for any r.v. $u_0\in H^1$, independent of  $\xi$, and for $l$ as in   \eqref{I8},  solution $u(t;u_0)$ satisfies 
$$
\lim_{t\to\infty} \EE \big(s_{3,l}( u(t;u_0) )\big) = -6 B_0 l +  o(l) = - 12 \eps^B l +  o(l) 
\quad \text{as} \quad \nu\to0.
$$

An easy calculation shows that if $u(t,x)$ is an $L$-periodic in $x$ solution of the equation in \eqref{B2}, where $\xi$
has the form \eqref{xi} with $e_s(x)$ replaced by  $e_s(x/L)$, then relation \eqref{bubu} still holds. 

\subsection{The Landau objection}\label{ss_10.2}

As we have already mentioned in Section \ref{ss_9.2}, Landau suggested a physical argument,
implying that a  relation  for a moment of velocity increments like relation  \eqref{univ} for the second moment, 
 may hold with a universal constant $C^K$, independent of  the  random force $\xi$,  only if the 
value of the moment, suggested by the relation, 
 is linear in the rate of energy dissipation $\eps$, like relation \eqref{K45} for the third moment. So the 2/3-law cannot hold 
 in the stronger  form  \eqref{univ} with a universal constant $C^K$. The goal of this section is to show that  for 
 1d turbulence, indeed, the only universal relation for the moments $S^{s}_{p,l}$    (see \eqref{S_skew})    is relation 
\eqref{bubu} for the cubic moment (which is linear in $\eps^B$). 

 Namely, for a stationary solution $u^{\nu \,st}(t,x)$ of the stochastic Burgers equation \eqref{B} and an integer $p\ge2$  consider the following
 hypothetic   relation for the  $p$-th moment of $u^{\nu \,st}$: 
\be\label{L}
 S_{p,l}^{s}( u^{\nu \,st}(t) )= C_* (\varepsilon^B l)^{q}+ o(\varepsilon^B l)^{q} \;\; \text{as} \; \nu\to0,
 \ee
where $l$ is any number from  the inertial range $[c_1\nu, c]$, and  $q >0$.
 We address the following question: for which 
$p$ and $q$ relation \eqref{L} holds with a \textit{universal} constant $C_*$, 
independent of   the random force $\xi$?

\begin{theorem}\label{t_L}
  If relation \eqref{L} holds for each random force $\xi$ as in \eqref{xi}, \eqref{B_assume},
    with a $C_*$ independent of  $\xi$,  then
$$p=3,~q=1,~C_*=-12.$$
\end{theorem}

\begin{proof} Let us abbreviate $u^{\nu \,st}(t)$ to $u(t)$. We  take some real number 
 $\mu >1$ and define $\tilde{\xi}(\tau):=\mu ^{-\frac{1}{2}}\xi(\mu \tau)$. This also is a process as in \eqref{xi} (with another set of
 independent Wiener processes $\beta_s$).  
   Denote $w(\tau,x):=\mu\, u(\mu \tau,x).$ Then $w$ is a stationary 
 solution of the equation
\begin{eqnarray}\label{L1}
w_{\tau}(\tau,x)+w(\tau,x)w_{x}(\tau,x)- {\nu^\mu} w_{xx}(\tau,x)=\mu^\frac{3}{2} \partial_{\tau}\tilde{\xi}(\tau,x), 
\qquad {\nu^\mu}=\nu \mu,
\end{eqnarray}
which is eq. \eqref{B}, \eqref{xi}  with scaled $\nu$ and $\xi$.

Consider  inertial range $J^1 = [c_1\nu, c]$ for eq.  \eqref{B} and   inertial
 range  $J^\mu = [c^\mu_1\nu, c^\mu]$  for eq.~\eqref{L1}. For small $\nu$ their intersection 
 $ J =J(\nu):= J^1\cap J^\mu$ is not empty. For any $l\in J$ relation  \eqref{L} holds for $u$ which solves eq.~\eqref{B} and for 
 $w$, solving  eq.~\eqref{L1}. Since 
 $
 S^{s}_{p,l}(w) = \mu^p S^{s}_{p,l} (u)
 $
 and as by \eqref{beB} 
  $\eps^B_w = \mu^3 \eps^B_u$, then from here we get that 
 $$
 \mu^p  C_* \big( \eps_u^B l\big)^q + o(\eps_u^B l )^q=   C_* \big( \mu^3 \eps_u^B l\big)^q + o(\eps_u^B l )^q
 $$
 for  all small $\nu>0$ and all $l\in J(\nu)$.
  As $\mu>1$, then by this equality   $q=p/3$.\,\footnote{
 This is in line with relation $| u(t,x+r) -u(t,x)| \asymp (\eps |r|)^{1/3}$ which appears in the theory of
 turbulence due to a basic dimension argument, without any relation to the equations, describing  the fluid.
 See \cite[(32,1)]{LL}. } 
 On the other hand, it follows from Theorem~\ref{t4.3} if $p$ is even and from relation \eqref{B45} and a discussion after it  if $p$ is odd 
  that $ |S^{s}_{p,l}( u)| \sim |l|$ for any integer $p\ge2$. Thus in \eqref{L} $q=1$, and so
 $p=3q=3$. Then by Theorem~\ref{t_45str}   $C_*=-12$ 	and the theorem is proved.
\end{proof}

\begin{remark}
1) The result of Theorem \ref{t_L} remains true with the same proof if relation \eqref{L} is claimed to 
hold not for all $l$ from the inertial range, but only  for $l$ from a strongly 
inertial range as in \eqref{I8}. In this form asymptotic  \eqref{L}  with $p=3$ and $q=1$ indeed is valid by Theorem~\ref{t_45str}.

2) We do not know if for some integer $p\ge2$, different from 3, asymptotical expansion for $S_{p,l}^{s}( u^{\nu \,st}(t) )$ of the form 
 \eqref{L}, valid for all $l$ from the inertial  range (or from a strongly inertial  range)  may hold  with a constant $C_*$ which  depends on the random force 
$\xi$. 
\end{remark}

\section {Inviscid   1d turbulence}\label{s_9}
\setcounter{equation}{0}
 
 In this section we study the  asymptotics of solutions for equation  (\ref{B2}) as $\nu\to 0$, define the limiting 
{\it entropy solutions,}
corresponding  to  $\nu=0$, and  establish their properties.

 \subsection{ Asymptotics of solutions as $\nu\to 0$}
 
For $m\in\N\cup 0$ we  denote by $\ov H^m$ the Sobolev space of order $m$
of  functions on $S^1$ with {\bf any} mean value and set 
define $\ov X^m_T=C(0,T;\ov H^m)$.  In (\ref{B2}) we  considered
 the problem 
  \be\la{B22}
\left\{\ba{rcl}
u_t(t,x)+u u_x-\nu u_{xx}&=&\eta(t,x)=\pa_t \xi(t,x),\qquad t\ge 0,
\\\
u(0,x)&=&u_0(x)
\ea\right|,\,\,x\in S^1.
\ee
Now let us also consider  another one
\be\la{B23}
\left\{\ba{rcl}
\vp_t(t,x)+\vp_x^2/2-\nu \vp_{xx}&=&\eta(t,x)
=\pa_t \zeta(t,x),\qquad t\ge 0,
\\\
\vp(0,x)&=&\vp_0(x)
\ea\right|,\,\,x\in S^1.
\ee
For solutions for  the latter problem the mean value 
$\ds\int \vp(x)dx$ is {\bf not} an integral of motion.
But obviously, if  $\vp(t,x)$ solves (\ref{B23}), then $u(t,x)=\vp_x(t,x)$ has zero mean-value and 
 solves (\ref{B22}) with 
$\xi=\pa_x\zeta$ and $u_0=\pa_x\vp_0$.
Conversely, if $u(t,x)$ is a solution of  (\ref{B22}), then $\vp(t,x)=\ds\int_0^x u(t,y)dy-\theta(t)$
with a suitable $\theta(t)$ (which is explicit in terms of $u$)
solves (\ref{B23}) with $\vp_0=\int_0^x u_0(y)dy$ and 
$\zeta(t,x)=\ds\int_0^x\xi(t,y)dy$.
Obviously,
$$
\vp \in \ov X_T^m \Leftrightarrow
 u \in  X_T^{m-1},\;\;\;
\vp_0 \in \ov H^m \Leftrightarrow
 u_0 \in  H^{m-1},\;\;\;
\zeta \in  {\ov X}^m_T \Leftrightarrow
  \xi \in  {\ov X}^{m-1}_T.
$$
So, essentially, (\ref{B22}) and (\ref{B23}) is the same problem. As we will now show, this  isomorphism
between the two problems  is an instrumental  tool to study the asymptotics of solutions for (\ref{B22}) as $\nu\to 0$.

We need a version of  the Oleinik estimates, valid up to $t=0$,  whose proof is similar  to that of  Theorem~\ref{t3.1}
(see in \cite{BK}):

\bt\la{t5.1}
Let $u$ solves (\ref{B22}) with $\xi\in { X}^{4}_T$
and $u_0\in H^2$. Then
$$
\sup_{0\le t\le T}|u(t)|_\infty\le B,\qquad 
\sup_{0\le t\le T}|u_x(t)|_1\le 2B,\qquad 
\sup_{0\le t\le T}|u_x^+(t)|_\infty\le B,
$$
where
\be\la{5.3w}
B=B_T(u_0,\xi)=
\max\big[|u_{0x}^+|_\infty+|\xi|_{X_T^4},
4|\xi|_{X_T^4}+|\xi|_{X_T^4}^{1/2}\big].
\ee
\et

 Let $u^\nu(t,x)$ solves 
(\ref{B22}) with 
\be\la{u0X}
u_0\in H^2,\qquad \xi\in {X}_T^4,\qquad 0 < \nu\le 1,
\ee
 let $\vp^\nu(t,x)$ be a corresponding solution of    problem 
(\ref{B23}) and $B$ be as in  (\ref{5.3w}).
Then Theorem~\ref{t5.1} implies that
\be\la{5.4}
|\vp^\nu_x(t)|_\infty = |u^\nu(t)|_\infty \le B\,, \qquad 
|\vp^\nu_{xx}(t)|_1 = |u^\nu_x(t)|_1 \le 2B,\qquad
|\vp^{\nu\,+}_{xx}(t)|_\infty = |u^{\nu\,+}_x(t)|_\infty \le B,
\ee
for any $0\le t\le T$.

\bt\la{t5.2}(S. Kruzkov)
Let $0<\nu_1<\nu_2\le 1$,  $1\le p<\infty$,  $T<\infty$ and  (\ref{u0X}) holds.  
 Let $u^\nu$ solves (\ref{B22}). Then
\be\la{5.5}
|u^{\nu_2}(t)-u^{\nu_1}(t)|_p\le C_p B^{1-\al_p}\ov\nu^{\al_p} e^{Bt\al_p},\qquad 0\le t\le T,
\ee
where $\ov\nu=\nu_2-\nu_1>0$, 
$\al_p=\min(\fr14,\fr1{3p})$, and $B=B_T(u_0,\xi)$
is defined in (\ref{5.3w}).

\et
{\bf Sketch of the proof.} Let $u^\nu(t;u_0,\xi)\in X_T^4$
be a solution of (\ref{B22}), and 
$\vp^\nu(t,x)\in X_T^5$
be the corresponding  solution of (\ref{B23}).
Denote $b(t,x)=\vp^{\nu_1}(t,x)-\vp^{\nu_2}(t,x)$.
Then subtracting  the equation for $\vp^{\nu_1}$ from that for  $\vp^{\nu_2}$ we get that 
$$
b_t+(\vp^{\nu_1}_x+\vp^{\nu_2}_x)b_x/2
=\nu_1\vp^{\nu_1}_{xx}-\nu_2\vp^{\nu_2}_{xx}.
\qquad
$$
Denote $E(t)=|b(t)|_2^2$. In view of \eqref{5.4} and 
 the equation above,
$$
( d/{dt}) E\le B E(t)+8B^2\ov\nu,\qquad E(0)=0,
$$
(see \ci{BK}) for the calculation). So by  Gronwall's 
inequality,
\be\la{5.7}
|b(t)|_2^2=E(t)\le 8B \ov\nu e^{Bt}.
\ee
Since 
$
\fr\pa{\pa x}b(t,x)
=u^{\nu_1}(t,x)-u^{\nu_2}(t,x) ,
$
 then by (\ref{5.4}) we have
\be\la{b*}
|b(t)|_{2,1}\le 4B
\ee
(we recall \eqref{norm}). Now the Gagliardo--Nirenberg inequality and (\ref{5.7}), (\ref{b*}) imply that 
$$
|u^{\nu_1}(t)-u^{\nu_2}(t)|_{4/3}=|\pa_x b(t,\cdot)|_{4/3}
\le
| b(t,\cdot)|_{2,1}^{1/2}| b(t,\cdot)|_2^{1/2}
\le C_1B^{1/2}(B^{1/4}e^{Bt/4}\ov\nu^{1/4}).
$$
This proves (\ref{5.5}) with $p\le 4/3$. To get
(\ref{5.5}) for $4/3\le p<\infty$, we apply the Riesz--Thorin
interpolation inequality to 
$v=u^{\nu_1}(t)-u^{\nu_2}(t)$ to get that 
$$
|v|_p\le |v|_\infty^{1-\fr43p}
 |v|_{4/3}^{\fr43p},\qquad p\ge \tfrac43,
$$
where   $|v|_\infty\le 2B$ by (\ref{5.4}). This complets the proof. 
 $\hfill\Box$

Inequalities  (\ref{5.5}) mean that for each $p<\infty$  the mapping \ 
$
(0,1]\ni\nu\mapsto u^\nu(t)\in C(0,T;L_p) \ 
$
is Cauchy-continuous as $\nu\to 0$. So, there exists an 
$
u^0\in\cap_{p<\infty}C(0,T;L_p),
$
such that
\be\la{5.9}
u^\nu
\underset{\nu\to 0}{-\!\!\!-\!\!\!-\!\!\!\to} \,\,u^0
\,\,{\rm in}\,\,C(0,T;L_p),\qquad \forall p<\infty.
\ee
Passing to  limits in the last estimate in (\ref{5.4})
and in (\ref{5.5}), we get:

\bc\la{c5.3} There exists $u^0(t,x) \in\cap_{p<\infty}C(0,T;L_p)$  such that
(\ref{5.9}) holds and (\ref{5.5}) stays true  for $\nu_1=0$ and $0<\nu_2\le1$. Moreover, 
\be\la{5.10}
|u^0|_{C(0,T;L_p)}\le B,\qquad \forall p<\infty.
\ee
\ec

Take any $t\in[0, T]$. Then $u^\nu(t)\underset{\nu\to 0}{-\!\!\!-\!\!\!-\!\!\!\to} u^0(t)$
in $L_1$, and hence, 
$
u^{\nu_j}(t,x)\underset{\nu_j\to 0}{-\!\!\!-\!\!\!-\!\!\!\to} u^0(t,x)$
for a.a. $x\in S^1$.
Therefore by (\ref{5.4}) we  obtain that also  
$
|u^0(t)|_\infty\le B$  for all  $t\le T.$

\subsection{The entropy solutions}
Similarly to  (\ref{M}) and  (\ref{Mt}), for $p<\infty$ 
we define the mapping
\be\la{M0}
\cM^0: H^2\times{X}_T^4\to 
C(0,T;L_p),\quad 
(u_0,\xi)\mapsto u^0(\cdot;u_0,\xi),
\ee
and for  $0\le t\le T$ -- the mappings
\be\la{Mt0}
\cM_t^0: H^2\times{X}_T^4\to L_p,
\quad 
(u_0,\xi)\mapsto u^0(t;u_0,\xi).
\ee
They are the limits of continuous  mappings (\ref{M}) and  (\ref{Mt})
as $\nu\to 0$, where we naturally embedded $X^2_T$ to $C(0,T; L_p)$ and $H^2$ to $L_p$. 
As the convergences (\ref{5.9}) are uniform on  bounded sets (since their rates depend only on $B$), 
 then the mappings $\cM^0$ and $\cM_t^0$ also are continuous.
\smallskip

Consider   equation (\ref{B2}) with $\nu=0$:
\be\la{B0}
u_t(t,x)+\tfrac12\pa_xu^2=\pa_t\xi(t,x),\qquad u(0,x)=u_0^\om(x).
\ee
It follows immediately from   (\ref{5.9}) that $u^0(t;u_0,\xi)$ with $u_0,\xi$ as above, solves (\ref{B0})
in the sense of generalized functions. 
 A generalized solution of (\ref{B0})
{\bf is not unique}, and the construction above single out among various solutions 
 a  {\bf unique} one.  It is called an {\it entropy}, or an {\it inviscid} solution of (\ref{B0}), e.g. see in  \cite{Daf}. 

Now let $\xi$ be the  random force (\ref{xi}).   Let $u_0\in H^2$ be a r.v., independent of  $\xi$.
\bd
$u^{0\om}(t,x;u_0,\xi):=\cM^0(u_0^\om,\xi^w)$
is an entropy solution for  problem (\ref{B0}), \eqref{xi}.
\ed

We will usually write a solution $u^0$  in this definition as $u^{0}(t,x;u_0)$ or $u^{0}(t;u_0)$.

Let $u_0\in H^2$, $\theta>0$,  $1\le p <\infty$ and $a>0$. 
Then Theorem \ref{t3.2},  convergence (\ref{5.9}) and  Fatou's lemma imply that 
\be\la{5.12}
\aE|u^0(t ; u_0)|_p^a\le C(a,B_4)\theta^{-a}, \qquad \forall\,t\ge\theta.
\ee

Due to convergence \eqref{5.9} with $p=1$, the mappings
$H^2\ni
u_0 \mapsto \cM_t^0(u_0, \xi)$, \, $t\ge0,
$
with a fixed $\xi \in {X}^4_T$ inherit estimate \eqref{nexp} and extend by continuity to  1-Lipschitz mappings  $L_1\to L_1$. Accordingly entropy solutions  $u^0(t;u_0)$ extend  to a Markov process in $L_1$. The latter is mixing: 
 there is a measure 
 $\mu_0\in \cP(L_1)$, satisfying $\mu_0( \cap_q L_q)=1$, such that for any r.v. $u_0 \in L_1$, independent of  $\xi$, 
 $$
 \cD u^0(t;u_0) \strela \mu_0 \quad \text{in} \quad \cP(L_p) \quad \text{as}\;\; t\to \infty,
 $$
 for any $p<\infty$. If $\cD u_0 = \mu_0$, then $u^{0\, st} (t) := u^0(t;u_0)$ is a stationary entropy solution,
  $\cD u^{0\, st} (t)  \equiv \mu_0$. Moreover, the viscous stationary measures $\mu_\nu$ as in Theorem~\ref{t4.55} weakly 
 converge, as $\nu\to0$,   to $\mu_0$:
  \be\label{entr_mix}
 \mu_\nu \strela \mu_0\quad\text{as}\quad \nu\to0, \quad \text{on each space } \cP(L_p), \;\; p<\infty.
  \ee
 See \cite[Chapter~8.5]{BK}. 
 
  The limiting ``entropy'' Markov process in $L_1$ admits an elegant presentation in terms of stochastic Lagrangians, e.g. 
see \ci{WKMS} and \ci{IK2003}.


\subsection{Moments of small-scale increments and  energy spectra of entropy solutions.}
Similar to Section \ref{ss_9.2}  we define the structure function of an entropy  solution $u^0(t,x)$ as 
$
   S_{p,l} (u^0) = \llangle  s_{p,l} (u^0) \rrangle,$ where  $s_{p,l} (v)= \ds\int |v(x+l) -v(x)|^pdx.
 $
 
By Theorem \ref{t4.3},  for suitable $C_1,c, c_*>0$, for any $0< \nu\le c_*$
and for every $p>0$ we have
\be\la{ine-dis}
S_{p,l}\big(u^\nu(\cdot ;u_0)\big) = \llangle s_{p,l} (u^\nu(\cdot,\cdot;u_0)\rrangle
\sim|l|^{\min(1,p)} \quad \mbox{if }\,\,|l|\in[C_1\nu,c].
\ee
Since functional $s_{p,l} $ is continuous on the space $L_{\max(1,p)}$ and $|s_{p,l} (v)| \le C_p |v|_{\max(1,p)}^p$, 
 then   convergence  \eqref{5.9}  
and the estimate
$$
\aE|u^\nu(t)|_p^a\le C(a,\theta)\qquad\forall\nu > 0,\,\,\,\forall
t\ge\theta>0, \,\,\,\forall a >0,
$$
which follows from Theorem \ref{t3.2}.\,ii), allow to  pass to a limit in (\ref{ine-dis}) as $\nu\to0$, $\nu \le |l| /C_1$, 
 and prove the following result.
\bt\la{t5.5}
Let 
$c$ be as in (\ref{ine-dis}). Then for any $u_0\in H^2$   entropy solution $u^0(t;u_0)$ of (\ref{B0}) satisfies
\be\la{*5}
S_{p,l}\big(   u^0(\cdot;u_0)\big)\sim |l|^{\min(1,p)},\qquad \forall p>0,
\ee 
for  $|l|\le c$.
\et

Since  (\ref{*5}) holds for all $|l|\le c$, then  for entropy solutions there is no dissipation range!
\medskip

 Now let us turn to the 4/5-law \eqref{K45}.   Consider relation \eqref{I5}.  Its l.h.s. equals 
 $ \lan s_{3,l}, \mu_\nu\ran$. The functional $ s_{3,l}$ is continuous on $L_3$,
  and by \eqref{entr_mix},
 $\mu_\nu\strela \mu_0$ in $\cP(L_3)$, where $\mu_0$ is the stationary measure for the inviscid Burgers equation. So
 passing to a limit as $\nu\to0$ in relation \eqref{I7} we get that 
 $$
 \EE \big(s_{3,l}( u^{0\, st}(t) )\big) = \lan  s_{3,l},  \mu_0\ran =
  -6 B_0 l +  O(l^3),
   $$
   where $ u^{0\, st}(t)$ is the stationary entropy solution. This relation is a version of    the 4/5-law for inviscid  1d turbulence. 
   \medskip

Similarly, one can pass to a limit in the energy-spectrum 
Theorem~\ref{t4.4} and get
\bt\la{t5.6} For $M\ge M'$ as in Theorem \ref{t4.4} 
and any $u_0\in H^2$, the energy spectrum of entropy solution 
$u^0(t;u_0)$ satisfies
\be\la{Ek}
E_{\mathbf k}^B(u^0)\sim {\mathbf k}^{-2},\qquad {\mathbf k}\ge 1.
\ee
\et

 In 3d turbulence   no analogies of Theorems \ref{t5.5} and \ref{t5.6} are known. 
That is, for the moment of writing  inviscid 3d turbulence is missing.

\noindent
{\bf \large Acknowledgement.}  I thank  Alexandre Komech for
 helping to edit the LN  for my  on-line lecture-course on the 1d Burgers turbulence 
  to this paper.

\end{document}